%% file: main-letter.tex
\documentclass{article}
\usepackage[margin=1in]{geometry}
\input{preamble}
\input{macros}
\usepackage{amsmath,amssymb,amsfonts}

\title{Good quantum LDPC codes with linear time decoder \\
from lossless expanders}

\author[1,2]{Ting-Chun Lin}
\author[2]{Min-Hsiu Hsieh}

\affil[1]{University of California San Diego}
\affil[2]{Hon Hai (Foxconn) Research Institute}

\begin{document}

\sloppy

\maketitle

\begin{abstract}
  Quantum low-density parity-check (qLDPC) codes are quantum stabilizer codes where each stabilizer acts on a constant number of qubits and each qubit is acted on by a constant number of stabilizers.
  We study qLDPC codes constructed from balanced products and lossless expanders.
  We found that assuming the existence of 2-sided lossless expander graphs with free group action,
    the resulting qLDPC codes have constant rate, linear distance, and linear time decoders.
\end{abstract}


\section{Introduction}
\label{sec:introduction}

The work of quantum error correction begins with Shor's discovery of the 9-qubit code  \cite{shor1995scheme}, followed by the CSS construction \cite{calderbank1996good} \cite{steane1996error} that represents a quantum stabilizer code with two classical linear codes. The CSS construction allows us to translate many classical results into quantum, including the existence of good quantum code,  where `good' means the code has code dimension and distance proportional to the number of qubits.

However, one problem that remains is the existence of good quantum low-density parity check (LDPC) code, i.e., a quantum CSS code    where the two classical codes are LDPC. The classical LDPC code is a classical linear code with a sparse parity-check matrix.
It was known for a while that good classical LDPC code exists through a random construction \cite{gallager1962low}.
Quantum, however, does not allow a simple random construction. The simplest construction where stabilizers are chosen independently   will not satisfy the commuting conditions required between the stabilizers. That is, to satisfy the commuting condition we need some structure besides pure randomness.

How to satisfy the commuting condition between the stabilizers?
  The simplest way is by taking the Cartesian product of two classical codes
  \cite{tillich2013quantum}.
  This defines a quantum LDPC code and has constant rate, 
    but unfortunately only gives $\Omega(\sqrt{n})$ distance.
  There are several improvements of distance to $\Omega(\textnormal{polylog}(n) \sqrt{n})$ \cite{freedman2002z2, evra2020decodable, kaufman2020new}
  but it is very challenging to improve even by a polynomial factor.

It was until
  \cite{hastings2021fiber}
  that breaks the square root barrier and brings the code distance to $\Omega(n^{3/5}/\textnormal{polylog} (n))$.
  The new idea is to consider a more general product, the fiber product,
    which locally looks like Cartesian product, but have additional global features.
  Soon 
  \cite{panteleev2021quantum} and
  \cite{breuckmann2021balanced} further improve on the construction.
    The first paper uses random lift of a graph and increase the code distance to $\Omega(n/\log n)$
    The second paper provides an explicit construction
      and proposed another kind of product, the balanced product,
      that unifies the previous constructions.

Now, we have the structure, so it is time to incorporate randomness
  and try to prove the existence of good quantum LDPC codes and linear time decoders.

\subsection{Main results}

In this paper, we show two results.
First, we give a new construction of good quantum LDPC code assuming the existence of 2-sided lossless expanders with free group actions. Second, under the same assumption, we show the quantum LDPC code has a linear time decoder.

\subsection{Technical tools} \label{sec:technical-tools}

Our construction is built on two ingredients; namely, a two-dimensional graph structure and (one-dimensional) lossless expander graphs.

The two-dimensional graph is constructed by  taking the Cartesian product of two one-dimensional graphs with free group actions over the same group and take the quotient of the diagonal action. Because of the Cartesian product, the final graph has not only vertices and edges, but also squares.
This feature of having squares provides the necessary structure for quantum codes.
This two-dimensional graph structure has been used in different context with different names including balanced product \cite{breuckmann2021balanced},   G-lift \cite{panteleev2021asymptotically}, and left-right Cayley graph \cite{dinur2021locally}.

Now, we talk about the second ingredient, the lossless expanders. Having a two-dimensional graph structure is not enough to give a good qLDPC code, one needs some kind of expander property. If we look back at the historical development of the expander code, we found two different approaches. One construction uses the eigenvalue expanders    together with Tanner code construction    \cite{sipser1996expander}. The other construction uses the lossless expanders    \cite{sipser1996expander}. Here, in the case of qLDPC, there are also two approaches. Ref.~\cite{panteleev2021asymptotically} took the first path, and this paper took the second.

So, what are lossless expander graphs?
  Lossless expander graphs are the optimal kind of vertex expanders.
  Given a $w$-regular graph,
    for any subset of vertices $S$ with small enough size
    (less than a constant fraction of the total number of vertices),
    its neighboring vertices, $N(S)$, has size $|N(S)| \ge (1-\epsilon)w|S|$
    for some small $\epsilon$.
  This almost saturate an upper bound for $|N(S)|$, $|N(S)| \le w|S|$.
  More details is provided in Subsection~\ref{sec:lossless-expander-def}.

Lossless property has many implications.
  One of them is the unique expansion,
    which means most of the vertices in $N(S)$ are connected to only one vertex in $S$.
  The more refined statement is that a lossless expander looks like a tree (Lemma~\ref{lem:lossless-split}).
    This allows us to have tight control on the global code structure such as distance and linear time decoder.

\subsection{Outline}
In Sec.~\ref{sec:preliminary} we introduce definitions and tools
  including the definitions of quantum error correcting codes, 
  lossless expander graphs,
  balanced product construction, 
  and chain complex.
In Sec.~\ref{sec:distance-qLDPC} we construct quantum codes
  using balanced product and lossless expander graphs
  and show that the codes is a good qLDPC code.
In Sec.~\ref{sec:decoder} we construct a linear time decoder for the qLDPC code,
  show the correctness of the decoder,
  and analyze the running time.
In Sec.~\ref{sec:conclusion} we discuss our results and future directions.

\section{Preliminary} \label{sec:preliminary}
\subsection{Quantum error correcting codes}

Here, we review quantum error correcting codes and quantum low-density parity-check codes. Quantum error correcting codes are described through stabilizers. To discuss the stabilizers, we first review the Pauli group.
The Pauli group on 1 qubit $G_1$ is the matrix group generated by the 2 by 2 Pauli matrices
  $$X = \begin{pmatrix}0 & 1 \\ 1 & 0 \end{pmatrix},
  Y = \begin{pmatrix}0 & -i \\ i & 0 \end{pmatrix},
  Z = \begin{pmatrix}1 & 0 \\ 0 & -1 \end{pmatrix}.$$
More explicitly, $G_1 = \{\pm I, \pm i I, \pm X, \pm i X, \pm Y, \pm i Y, \pm Z, \pm i Z\}$,  where $I$ is the 2 by 2 identity matrix. 
The Pauli group on $n$ qubits $G_n = G_1^{\otimes n}$ is the $n$ fold tensor product of $G_1$.

\begin{definition}[Quantum stabilizer code]
A quantum stabilizer code $Q$ that encodes $k$ logical qubits into $n$ physical qubits
  is defined by specifying a stabilizer group $S$ of $(n-k)$ generators, 
  where $S$ is an Abelian subgroup of the Pauli group on $n$ qubits.
A codeword, $|\psi\>$, is a vector in $\CC^{2^n}$ such that
  the vector is invariant under the action of each element in $S$,
  $s|\psi\> = |\psi\>, \forall s \in S$. 
The codespace is the vector space spanned by the codewords.
\end{definition}

To describe $S$, it is enough to specify a generating set.
In the case of quantum CSS code,
  each stabilizer in the generating set contains only Pauli X or Pauli Z.
Therefore, we can explicitly describe $S$
  by giving two matrices $H_x \in \FF_2^{m_x\times n} \cong [\FF_2^{m_x} \rightarrow \FF_2^{n}], H_z \in \FF_2^{m_z\times n} \cong [\FF_2^{m_z} \rightarrow \FF_2^{n}]$,
  where $S$ is generated by 
  $m_x$ X-stabilizers $s_j = \prod_{i=1}^n X_i^{a_{j, i}}, j = 1, 2, ..., m_x$
  and $m_z$ Z-stabilizers $s_k = \prod_{i=1}^n Z_i^{b_{k, i}}, k = 1, 2, ..., m_z$, 
  where $X_i, Z_i$ are the Pauli X and Z operator acting on the $i$-th qubit
  and $a_{j, i}, b_{k, i}$ are the entries of $H_x$ and $H_z$.
Recall that to have a well defined quantum stabilizer code, $S$ has to be Abelian.
  This is equivalent to the condition $H_x H_z^T = 0$.
We denote the corresponding quantum CSS code as $Q(H_x, H_z)$.

Besides the stabilizers, another important object is the logical operators.
The logical operators are the Pauli strings that commute with all stabilizers.
  They are called logical operators because
  they map codewords to codewords.
The trivial logical operators are the logical operators
  that acts trivially on all codewords.
In fact, the trivial logical operators are exactly the stabilizers.

Mathematically, 
  the logical X-operators are the kernel of $H_z^T$,
  $\Ker H_z^T = \{a \in \FF_2^n: H_z^T a^T = 0\}$,
  the trivial logical X-operators are the image of $H_x$,
  $\Ima H_x = \{v H_x \in \FF_2^n: v \in \FF_2^{m_x}\}$.

Below are some important parameters of the code.
The weight of a vector is the number of nonzero entries denoted as $|v|$.

\begin{itemize}
  \item The length of the code: $n$.
  \item The dimension of the code: $k$.
  \item The distance of the code: $d$, which is the minimal weight of all nontrivial logical operators.
  \item The weight of the code: $w$, which is the maximal weight of all column and row vectors in $H_x$ and $H_z$.
\end{itemize}

Because the code is a CSS code, we can write $d = \min(d_x, d_z)$, where $d_x$ and $d_z$ are the lengths of the shortest nontrivial logical X-operator and Z-operator.


Quantum low-density parity-check (LDPC) codes are codes with constant weight, $w = \Theta(1)$. 
This implies each stabilizer only acts on a constant number of qubits and each qubit is only acted by a constant number of stabilizers.
Good quantum LDPC codes are quantum LDPC codes that further have linear dimension and linear distance, $k = \Theta(n)$ and $d = \Theta(n)$.

\subsection{Graphs}

Here, we review the definition of bipartite graphs, adjacency matrices, regular graphs, Cayley graphs, and lossless expander graphs.

We use $\Xi \equiv (V_0, V_1, E)$ to denote a bipartite graph, where $V_0$, $V_1$ are the sets of vertices on the 2 sides, and $E \subseteq V_0 \times V_1$ is the set of edges between the vertices. We use variables $\nu_0, \nu_1$ to denote subsets of $V_0, V_1$, and use variables $x_0, x_1$ to denote individual vertices in $V_0, V_1$. We use $E(\nu_0, \nu_1)$ to denote the set of edges between $\nu_0$ and $\nu_1$.

The neighbors of a vertex $x_0$ within $\nu_1$ is denoted as $N_{\nu_1}(x_0)$. 
We abbreviate $N_{V_1}(x_0)$ to $N_1(x_0)$ and $N_{V_0}(x_1)$ to $N_0(x_1)$.
Similarly, the neighbors of a subset $\nu_0$ within $\nu_1$ is denoted as $N_{\nu_1}(\nu_0)$.

The degree of a vertex $x_0$ in $\nu_1$ is defined as the size of the neighbors of $x_0$ in $\nu_1$, i.e. $\deg_{\nu_1}(x_0) \coloneqq |N_{\nu_1}(x_0)|$.

The adjacency matrix of a bipartite graph 
  is a matrix $L(E) \in \FF_2^{V_0 \times V_1} \cong [\FF_2^{V_0} \rightarrow \FF_2^{V_1}]$,
  where $L(E)_{x_0, x_1} = 1$ if $(x_0, x_1) \in E$,
  otherwise $L(E)_{x_0, x_1} = 0$.
Equivalently,
  $L(E) e_{x_0} = \sum_{x_1 \in N_{V_1}(x_0)} e_{x_1}$,
  where $e_{x_0}$ and $e_{x_1}$ are the basis vectors in $\FF_2^{V_0}$ and $\FF_2^{V_1}$.

A bipartite graph $\Xi$ is $(w_0, w_1)$-regular if the degrees of all vertices in $V_0$ are equal to $w_0$, and the degrees of the all vertices in $V_1$ are equal to $w_1$. Notice that $w_0$ and $w_1$ are the weights of the column and row vectors of the adjacency matrix.

\subsubsection{Cayley graphs and graphs with free group action}

Here, we discuss graphs with free group action, which is crucial for the balanced product construction. Cayley graph is the key example for a graph with free group action. In fact, all graphs with free group action can be decomposed into Cayley graphs.

A bipartite graph $\Xi$ is $G$-invariant if there exist $G$-actions on $V_0$ and $V_1$, such that if $(x_0, x_1) \in E$, then $(gx_0, gx_1) \in E$.
We also say the graph has $G$-symmetry.
Later we only consider a special case of $G$-action, where the action is free.
A group action on a set $V$ is free, if for all $x \in V$, $gx = x$ implies $g = 1$.
A group action on $\Xi$ is free, if the actions on both $V_0$ and $V_1$ are free.

The left (acting) bipartite Cayley graph,
  $\Gamma_{\textnormal{left}}(G, A) = (V_0, V_1, E)$, 
  is a bipartite graph constructed from a group $G$ 
  and a generating set $A \subseteq G$. 
The graph has vertices $V_0 = G, V_1 = G$, 
  and edges $E = \{(g, ag) : g \in G, a \in A\}$. 
We can also have the generating set acts from the right, 
  which defines the right bipartite Cayley graph, 
  $\Gamma_{\textnormal{right}}(G, B)$. 
The left (right) bipartite Cayley graph is $G$-invariant by the right (left) group action which acts freely.

\subsubsection{Lossless expander} \label{sec:lossless-expander-def}

A lossless expander graph is a regular graph where the vertex expansion is optimal i.e approximately equals to its degree. 

\begin{definition}[Small set vertex expansion]
  A bipartite graph $\Xi$ has $(c, \alpha)$-vertex expansion from $V_0$ to $V_1$ if for any subset $\nu_0 \subseteq V_0$ with $|\nu_0| < c |V_0|$, $|N_{V_1}(\nu_0)| \ge \alpha |\nu_0|$. 
\end{definition}

\begin{definition}[1-sided lossless expander]
  A $(w_0, w_1)$-regular bipartite graph $\Xi$ is a 1-sided $(c, \epsilon)$-lossless expander from $V_0$ to $V_1$, if it has $(c, (1-\epsilon)w_0)$-vertex expansion from $V_0$ to $V_1$.
\end{definition}

\begin{definition}[2-sided lossless expander] A $(w_0, w_1)$-regular bipartite graph $\Xi$ is a 2-sided $(c, \epsilon)$-lossless expander, if it has $(c, (1-\epsilon)w_0)$-vertex expansion from $V_0$ to $V_1$ and $(c, (1-\epsilon)w_1)$-vertex expansion from $V_1$ to $V_0$.
\end{definition}

It is known that 1-sided lossless expanders with free group actions exist \cite{capalbo2002randomness}. But 2-sided lossless expanders with free group actions are unknown at this moment.




\subsection{Balanced product construction}

Balanced product is used to guarantee 
  the commuting condition of the stabilizers.
It is obtained by first taking the Cartesian product, then taking the quotient over the diagonal group action.




We first review the definition of the hypergraph product \cite{tillich2013quantum}
  which is the same as the Cartesian product. 


\begin{definition}[Hypergraph product]
  Given two bipartite graphs $\Xi_X = (V_{X, 0}, V_{X, 1},E_X)$ and $\Xi_Y = (V_{Y, 0}, V_{Y, 1}, E_Y)$, the \emph{hypergraph product} of $\Xi_X$ and $\Xi_Y$, $\Xi_X \times \Xi_Y$, has 
  \begin{itemize}
    \item vertices: $V_{00} = V_{X,0} \times V_{Y,0},
      V_{10} = V_{X,1} \times V_{Y,0},
      V_{01} = V_{X,0} \times V_{Y,1},
      V_{11} = V_{X,1} \times V_{Y,1}$,
    \item edges: \\
      $E_{*0} = \{((x_0,y_0), (x_1,y_0)) : (x_0,x_1) \in E_X, y_0 \in V_{Y,0}\}, \\
      E_{*1} = \{((x_0,y_1), (x_1,y_1)) : (x_0,x_1) \in E_X, y_1 \in V_{Y,1}\}, \\
      E_{0*} = \{((x_0,y_0), (x_0,y_1)) : x_0 \in V_{X,0}, (y_0,y_1) \in E_Y\}, \\
      E_{1*} = \{((x_1,y_0), (x_1,y_1)) : x_1 \in V_{X,1}, (y_0,y_1) \in E_Y\}$,
    \item faces: $F = \{((x_0,y_0), (x_1,y_0), (x_0,y_1), (x_1,y_1)) : (x_0,x_1) \in E_X, (y_0,y_1) \in E_Y\}$.
  \end{itemize}
\end{definition}  

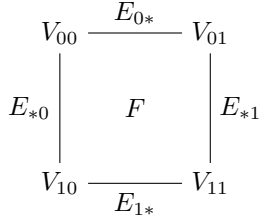
\begin{figure}
  \centering
  \begin{tikzpicture}
    \node (V00) at (0,0) {$V_{00}$};
\node (V10) at (0,-2) {$V_{10}$};
\node (V01) at (2,0) {$V_{01}$};
\node (V11) at (2,-2) {$V_{11}$};

\draw (V00) to node[auto, swap] {$E_{*0}$} (V10);
\draw (V00) to node[auto] {$E_{0*}$} (V01);
\draw (V01) to node[auto] {$E_{*1}$} (V11);
\draw (V10) to node[auto, swap] {$E_{1*}$} (V11);

\node (F) at (1,-1) {$F$};
  \end{tikzpicture}
  \caption{Balanced product of bipartite graphs.}
  \label{fig:hypergraph_BG}
\end{figure}

When the graphs $\Xi_X$ and $\Xi_Y$ have $G$-action, the hypergraph product $\Xi_X \times \Xi_Y$ has a $G$-action defined by the diagonal $G$-action. After quotienting this action, we obtain the balanced product.

\begin{definition}[Balanced product]
  The \emph{balanced product} of two bipartite graphs $\Xi_X, \Xi_Y$ with free $G$-action, denoted by $\Xi_X \times_G \Xi_Y$, 
    has 
  \begin{itemize}
    \item vertices: $V_{\a\b} \coloneqq V_{X,\a} \times V_{Y,\b} / G,$
      for $\a, \b \in \{0, 1\}$,
      where $(x_\a, y_\b) \sim (gx_\a, gy_\b)$,
    \item edges: for $\a, \b \in \{0, 1\}$,
      $E_{*\b} \coloneqq \{((x_0,y_\b), (x_1,y_\b)) : (x_0,x_1) \in E_X, y_\b \in V_{Y,\b}\} / G$, 
      where $((x_0,y_\b), (x_1,y_\b)) \sim ((gx_0,gy_\b), (gx_1,gy_\b))$, and 
      $E_{\a*} \coloneqq \{((x_\a,y_0), (x_\a,y_1)) : x_\a \in V_{X,\a}, (y_0,y_1) \in E_Y\} / G$, 
      where $((x_\a,y_0), (x_\a,y_1)) \sim ((gx_\a,gy_0), (gx_\a,gy_1))$,
    \item faces:
      $F \coloneqq \{((x_0,y_0), (x_1,y_0), (x_0,y_1), (x_1,y_1)) : (x_0,x_1) \in E_X, (y_0,y_1) \in E_Y\} / G$, 
     where $((x_0,y_0), (x_1,y_0), (x_0,y_1), (x_1,y_1)) \sim ((gx_0,gy_0), (gx_1,gy_0), (gx_0,gy_1), (gx_1,gy_1))$.
  \end{itemize}
\end{definition}

The balanced product of bipartite graphs is illustrated in 
  Figure~\ref{fig:hypergraph_BG}.
We use $(V^*, E^*, F)$ as a shorthand for
  $(V_{00}, V_{10}, V_{01}, V_{11}, E_{*0}, E_{*1}, E_{0*}, E_{1*}, F)$.

An important example of a balanced product graph is the left-right Cayley graph. 

\begin{definition}[Left-right Cayley graph] \label{def:left-right-cayley-graph}
  
  The left-right bipartite Cayley graph $\Gamma_2(G, A, B) \coloneqq \Xi_X \times_G \Xi_Y$
  where $\Xi_X \coloneqq\Gamma_{\rm right}(G, A^{-1}), \Xi_Y\coloneqq\Gamma_{\rm right}(G, B)$
  are right bipartite Cayley graphs with $A^{-1} = \{a^{-1} : a \in A\}$.

  Explicitly, the graph has
  \begin{itemize}
    \item vertices: $V_{00} \cong V_{10} \cong V_{01} \cong V_{11} \cong G \times G / G \cong G$,
    \item edges: 
      \begin{align*}
        E_{*0} &= \{(g, ag) : g \in G, a \in A\}, \\
        E_{*1} &= \{(gb, agb) : gb \in G, a \in A\}, \\
        E_{0*} &= \{(g, gb) : g \in G, b \in B\}, \\
        E_{1*} &= \{(ag, agb) : ag \in G, b \in B\},
      \end{align*}
    \item faces: $\{(g, ag, gb, agb): g \in G, a \in A, b \in B\}$,
  \end{itemize}
  where $G \times G / G \cong G$
  uses the bijection $[(x_\a, y_\b)] \mapsto x_\a^{-1} y_\b$
  for $\a, \b \in \{0,1\}$.
  $[(x_\a, y_\b)]$ denotes the equivalent class of $(x_\a, y_\b)$ in $G \times G / G$.
\end{definition}

\begin{figure}
  \centering
  \begin{tikzpicture}
    \node (V00) at (0,0) {$V_{00}$};
\node (V10) at (0,-2) {$V_{10}$};
\node (V01) at (2,0) {$V_{01}$};
\node (V11) at (2,-2) {$V_{11}$};

\draw (V00) to node[auto, swap] {$E_{*0}$} (V10);
\draw (V00) to node[auto] {$E_{0*}$} (V01);
\draw (V01) to node[auto, swap] {$E_{*1}$} (V11);
\draw (V10) to node[auto] {$E_{1*}$} (V11);

\draw (0,0)+(0.7,-0.7) node (g) {$g$};
\draw (0,-2)+(0.7,-0.7) node (ag) {$ag$};
\draw (2,0)+(0.7,-0.7) node (gb) {$gb$};
\draw (2,-2)+(0.7,-0.7) node (agb) {$agb$};

\path (g) -- node [sloped] {$\ni$} (V00);
\path (ag) -- node [sloped] {$\ni$} (V10);
\path (gb) -- node [sloped] {$\ni$} (V01);
\path (agb) -- node [sloped] {$\ni$} (V11);

\draw (g) to (ag);
\draw (g) to (gb);
\draw (gb) to (agb);
\draw (ag) to (agb);
  \end{tikzpicture}
  \caption{Left-right bipartite Cayley graph.}
  \label{fig:left-right-cayley-graph}
\end{figure}
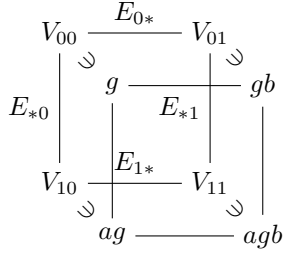

The left-right bipartite Cayley graph is illustrated in 
  Figure~\ref{fig:left-right-cayley-graph}.Note that we labeled the edges to point out the four vertices, $g \in V_{00}, ag \in V_{10}, gb \in V_{01}, agb \in V_{11}$, form a square. The square appears because the left action commutes with the right action. The existence of squares will be used later in the proof.

Back to the balanced product, if $z_{00} \in V_{00}, z_{10} \in V_{10}, z_{01} \in V_{01}$
  and they are adjacent $(z_{00}, z_{10}) \in E_{*0}$, $(z_{00}, z_{01}) \in E_{0*}$,
  then there exists a unique vertex $z_{11} \in V_{11}$ that completes the square 
  $(z_{00}, z_{10}, z_{01}, z_{11}) \in F$.
  We denote such vertex as $z_{10} \times_{z_{00}} z_{01}$.
This is proven in the Lemma~\ref{lem:square-completion} and is used to show the balanced product gives a chain complex.


\subsection{Chain complex} \label{sec:chain-complex}

Here, we introduce the language of chain complexes from homological algebra. The connection between quantum error correcting codes and chain complexes was first discussed in \cite{bravyi2014homological}. 

As we will see, homological algebra gives a natural language to discuss quantum CSS codes,
  which drives the recent breakthrough in qLDPC codes
    \cite{hastings2021fiber}
  by drawing analogies between codes, chain complexes and manifolds.

\begin{definition}[Chain complex]
A \emph{chain complex} $\mathcal{C}$ is a sequence of vector spaces, $C_i$, together with linear maps, $\partial_i: C_i \rightarrow C_{i-1}$ called the \emph{boundary operators}, where these boundary operators satisfy
  \begin{equation}
    \partial_{i-1} \partial_i  = 0.
  \end{equation}
\end{definition}

The kernel and the image of a boundary operator is defined as
$\Ker \partial_i \coloneqq \{c_i \in C_i : \partial_i c_i = 0\}$,
$\Ima \partial_i \coloneqq \{\partial_i c_i \in C_{i-1} : c_i \in C_i\}$.

Notice the constraint $\partial_{i-1} \partial_i = 0$ is similar
  to the the commuting constraint for CSS codes $H_xH_z^T=0$.
Indeed, a quantum CSS code is equivalent to a 3-term chain complex:
\begin{equation}
  \cC: [C_2 \xrightarrow{\partial_2} C_1 \xrightarrow{\partial_1} C_0]
  \equiv [\FF_2^{m_z} \xrightarrow{H_z^T} \FF_2^{n} \xrightarrow{H_x} \FF_2^{m_x}].
\end{equation}

Using this equivalence, we can rewrite the objects in quantum codes
  in the language of chain complexes.

\begin{itemize}
  \item Logical Z-operators $\equiv \Ker \partial_1$.
  \item Trivial logical Z-operators $\equiv \Ima \partial_2$.
  \item Z-distance $d_z = d_1(\cC) = \min_{v \in \left(\Ker \partial_1 - \Ima \partial_2 \right)} |v|$.
\end{itemize}

Similarly for X-operators, we consider the \emph{dual chain complex},
\begin{equation}
  \cC^T: C_0 \xrightarrow{\partial_1^T} C_1 \xrightarrow{\partial_2^T} C_2.
\end{equation}

\begin{itemize}
  \item Logical X-operators $\equiv \Ker \partial_2^T$.
  \item Trivial logical X-operators $\equiv \Ima \partial_1^T$.
  \item X-distance $d_x = d_1(\cC^T) = \min_{v \in \left(\Ker \partial_2^T - \Ima \partial_1^T \right)} |v|$.
\end{itemize}

Note that distance is not an intrinsic property of a chain complex, instead, it is basis dependent. 

\section{Construct and prove good qLDPC} \label{sec:distance-qLDPC}
In this section, we will construct and prove the existence of good quantum LDPC codes  assuming the existence of 2-sided lossless expanders with free group actions.
The construction is based on the balanced product of lossless expander graphs. 

\begin{theorem}[Good quantum LDPC] \label{thm:good-qLDPC}
  Assume that the 2-sided lossless expander with free group action in Conjecture~\ref{con:2-sided-lossless-expander-with-symmetry} exists. Then for all $0 < r < 1$, there exist $\delta > 0$, $w \in \NN$
  and a construction of an infinite family of quantum error-correcting codes
    $\{C_i\}$ with parameters $[[n_i, k_i, d_i]]$,
  such that $n_i$ approaches infinity as $i$ increases, $k_i/n_i \ge r, d_i/n_i \ge \delta$
  and $C_i$ has weight $w$.
\end{theorem}

We first construct the code in Section~\ref{sec:construct-qLDPC}, and then prove that the code has constant rate, linear distance and finite weight in Section~\ref{sec:proof-qLDPC}.



\subsection{Construction of qLDPC} \label{sec:construct-qLDPC}

We first state the conjecture that is assumed for the entire paper.

\begin{conjecture}[2-sided lossless expander with free group action]
  \label{con:2-sided-lossless-expander-with-symmetry}
  There is a family of groups $\{G_i\}$ with $|G_i| \rightarrow \infty$,
  such that for any $\epsilon > 0$ and interval $I \subseteq (0,1)$,
    there exist parameters $(w_0, w_1), (c, \epsilon)$
      and large enough $i_0$ such that
    for all $G = G_i, i \ge i_0$
      there exists a bipartite graph $X = (V_0, V_1, E)$ such that

  \begin{itemize}
    \item $w_0/w_1 \in I$,
    \item $X$ is $(w_0, w_1)$-regular,
    \item $X$ is a 2-sided $(c, \epsilon)$-lossless expander,
    \item $X$ has free $G$-action,
    \item $|V_0| = \Theta(|G|)$.
  \end{itemize}
\end{conjecture}

Now, we introduce the notations used throughout the section.
The balanced product graph is constructed using two bipartite graphs
$X_{\updownarrow} = (V_{0*}, V_{1*}, E_{\updownarrow})$ and $X_{\leftrightarrow} = (V_{*0}, V_{*1}, E_{\leftrightarrow})$.
We denote the vertices and edges of the balanced product graph by
$V_{00}, V_{10}, V_{01}, V_{11}, E_{*0}, E_{*1}, E_{0*}, E_{1*}$
as shown in Fig.~\ref{fig:hypergraph_BG}.

The bipartite graphs $X_{\updownarrow}, X_{\leftrightarrow}$ are chosen to satisfy the conditions given in the following theorem.
This can be obtained assuming the conjecture~\ref{con:2-sided-lossless-expander-with-symmetry}.

\begin{theorem}[] \label{thm:2-sided-lossless-expander-with-symmetry-2d}
Assume Conjecture~\ref{con:2-sided-lossless-expander-with-symmetry} holds.  For any $\epsilon > 0$ and intervals $I_\updownarrow, I_\leftrightarrow \subseteq (0,1)$, there exist parameters $(w_\downarrow, w_\uparrow), (w_\rightarrow, w_\leftarrow),    (c_\updownarrow, \epsilon_\updownarrow), (c_\leftrightarrow, \epsilon_\leftrightarrow)$, such that
  \begin{itemize}
      \item $w_\downarrow/w_\uparrow \in I_\updownarrow$,
      \item $w_\rightarrow/w_\leftarrow \in I_\leftrightarrow$,
      \item $\epsilon_\updownarrow w_\downarrow \le w_\uparrow, \epsilon_\updownarrow w_\uparrow \le w_\downarrow$,
      \item $w_\uparrow \epsilon_\leftrightarrow, w_\downarrow \epsilon_\leftrightarrow \le \epsilon$,
      \item $\epsilon_\leftrightarrow \le \epsilon$ (this inequality is not necessary since it is implied from the last line), 
      \item $\epsilon_\updownarrow \le \epsilon$,
  \end{itemize}
  and for any $N \in \NN$,
  there exist group $G$ with $|G|>N$,
  and bipartite graphs $X_\updownarrow = (V_{0*}, V_{1*}, E_\updownarrow), X_\leftrightarrow = (V_{*0}, V_{*1}, E_\leftrightarrow)$
  such that 
  \begin{itemize}
      \item $X_\updownarrow$ is $(w_\downarrow, w_\uparrow)$-regular,
      \item $X_\leftrightarrow$ is $(w_\rightarrow, w_\leftarrow)$-regular,
      \item $X_\updownarrow$ is a 2-sided $(c_\updownarrow, \epsilon_\updownarrow)$-lossless expander,
      \item $X_\leftrightarrow$ is a 2-sided $(c_\leftrightarrow, \epsilon_\leftrightarrow)$-lossless expander,
      \item $X_\updownarrow, X_\leftrightarrow$ have free $G$-actions,
      \item $|V_{0*}| = \Theta(|G|)$ and $|V_{*0}| = \Theta(|G|)$.
  \end{itemize}
\end{theorem}

\begin{proof}
  We find $X_{\updownarrow}$ and $X_{\leftrightarrow}$ in order.
  Namely, we first find $w_\downarrow, w_\uparrow, c_\updownarrow, \epsilon_\updownarrow, i_0$ that satisfy
  the conjecture and 
  $w_\downarrow/w_\uparrow \in I_\updownarrow$, $\epsilon_\updownarrow \le \min(\epsilon, w_\uparrow/w_\downarrow, w_\downarrow/w_\uparrow)$.
  Then, we find $w_\rightarrow, w_\leftarrow, c_\leftrightarrow, \epsilon_\leftrightarrow, i'_0$ that satisfy
  the conjecture and 
  $w_\rightarrow/w_\leftarrow \in I_\leftrightarrow$, $\epsilon_\leftrightarrow \le \epsilon/\max(w_\uparrow, w_\downarrow)$. This automatically satisfies $\epsilon_\leftrightarrow \le \epsilon$.

  Now, we pick large enough $G = G_i$ such that $|G| > N$ and $i > i_0, i_0'$.

  Finally, we find $X_{\updownarrow}$ and $X_{\leftrightarrow}$ that are lossless expanders, have free $G$-action, 
    and satisfy linear size conditions $|V_{0*}| = \Theta(|G|)$, $|V_{*0}| = \Theta(|G|)$.
  These are possible because of the conjecture.
\end{proof}

Now, we are ready to construct the code in Theorem~~\ref{thm:good-qLDPC}.
\begin{itemize}
  \item By Theorem~\ref{thm:2-sided-lossless-expander-with-symmetry-2d}, there exists 
    $X_\updownarrow, X_\leftrightarrow$,
    such that the parameters satisfy
    \begin{itemize}
        \item $\frac{(w_\downarrow-w_\uparrow)(w_\leftarrow-w_\rightarrow)}
        {w_\downarrow w_\leftarrow + w_\uparrow w_\rightarrow} \ge r$,
        \item $|V_{0*}| = \Theta(|G|), |V_{*0}| = \Theta(|G|)$,
        \item $w_\uparrow \epsilon_\leftrightarrow, w_\downarrow \epsilon_\leftrightarrow, \epsilon_\leftrightarrow, \epsilon_\updownarrow \le \epsilon < 1/12$.
    \end{itemize}
  \item Next, we take the balanced product $X_\updownarrow \times_G X_\leftrightarrow$
    which gives a 3-term chain complex 
    $\FF_2^{V_{00}}
    \xrightarrow{\partial_2} \FF_2^{V_{10}} \oplus \FF_2^{V_{01}}
    \xrightarrow{\partial_1} \FF_2^{V_{11}}$,
    where $\partial_2(v_{00}) = (L(E_{*0})(v_{00}), L(E_{0*})(v_{00}))$ and $\partial_1((v_{10}, v_{01})) = L(E_{1*})(v_{10}) + L(E_{*1})(v_{01})$ are induced from the adjacency matrices of the balanced product graph,
    and $L(E_{*0}): \FF_2^{V_{00}} \rightarrow \FF_2^{V_{10}}$,
      $L(E_{0*}): \FF_2^{V_{00}} \rightarrow \FF_2^{V_{01}}$,
      $L(E_{*1}): \FF_2^{V_{01}} \rightarrow \FF_2^{V_{11}}$,
      $L(E_{1*}): \FF_2^{V_{10}} \rightarrow \FF_2^{V_{11}}$.
  \item Finally, we obtain a quantum CSS code
    $Q(H_x = \partial_1^T, H_z = \partial_2)$
    by taking $\partial_1^T$ and $\partial_2$ as the parity-check matrices
    where the vertices in $V_{00}$, $V_{10} \cup V_{01}$, and $V_{11}$
    represent the Z-stabilizers, the qubits, and the X-stabilizers.
\end{itemize}

We need 2-sided lossless expanders for our construction
  because we need to show both $d_z$ and $d_x$ are linear to $n$.
One side of the lossless expansion will show $d_z$ is linear to $n$
  and the other side will show $d_x$ is linear to $n$.

\subsection{Proof of Theorem~\ref{thm:good-qLDPC}} \label{sec:proof-qLDPC}

The main challenge for the proof is to show that the qLDPC has linear distance. This  relies on a series of lemmas in Appendix~\ref{sec:distance-lemma}.
The main work is done in what we called the small set LTC lemma which is proven in Appendix~\ref{sec:proof-key-lemmas}.

\begin{lemma}[Small set LTC] \label{lem:small-set-LTC}
  Consider the chain complex, $C_2 \xrightarrow{\partial_2} C_1 \xrightarrow{\partial_1} C_0$, 
    constructed from the balanced product graph
    $X^{\updownarrow} \times_G X^{\leftrightarrow} = (V^*, E^*, F)$,
  where $X^{\updownarrow}$ and $X^{\leftrightarrow}$ satisfy
  \begin{itemize}
    \item $X_\updownarrow$ is $(w_\downarrow, w_\uparrow)$-regular,
    \item $X_\leftrightarrow$ is $(w_\rightarrow, w_\leftarrow)$-regular,
    \item $X_\updownarrow$ is a 1-sided $(c_\updownarrow, \epsilon_\updownarrow)$-lossless expander from $V_{0*}$ to $V_{1*}$,
    \item $X_\leftrightarrow$ is a 1-sided $(c_\leftrightarrow, \epsilon_\leftrightarrow)$-lossless expander from $V_{*0}$ to $V_{*1}$,
    \item $X_\updownarrow, X_\leftrightarrow$ have free $G$-actions,
    \item $\epsilon_\updownarrow w_\downarrow \le w_\uparrow$,
    \item $w_\downarrow \epsilon_\leftrightarrow \le \epsilon$,
    \item $\epsilon_\leftrightarrow \le \epsilon$ (implied by the last line), 
    \item $\epsilon_\updownarrow \le \epsilon$.
  \end{itemize}

If $c_0 = \partial{c_1}$ for some short (normalized) locally minimal $c_1 = (v_{10}, v_{01}) \in C_1$, 
    with $|v_{10}| < \min (c_\updownarrow |V_{0*}|/w_\uparrow, c_\leftrightarrow |V_{*0}|)$ and $
    |v_{01}| < \min (c_\leftrightarrow |V_{*0}|/w_\leftarrow, c_\updownarrow |V_{0*}|)$,
  then
  \begin{equation}
    \left(\frac{1}{2} - 6 \epsilon\right)|c_1|_w \le |c_0|_w,
  \end{equation}
    where the normalized weight of the vectors $|c_1|_w = |v_{10}|/w_\downarrow + |v_{01}|/w_\rightarrow$
    and $|c_0|_w = |c_0|/(w_\downarrow w_\rightarrow)$. 
\end{lemma}


\begin{proof}[Proof of Theorem~\ref{thm:good-qLDPC}]
  Now, we prove the code constructed in the previous section
    has arbitrarily large code length, constant rate, linear distance and finite weight.
  
  \begin{proof}[Proof of arbitrarily large code length]
    The code length $n = |V_{01}| + |V_{10}| = |V_{0*}||V_{*1}|/|G| + |V_{1*}||V_{*0}|/|G| = \Theta(|G|)$.
    By Theorem~\ref{thm:2-sided-lossless-expander-with-symmetry-2d},
    $|G|$ can be arbitrarily large.
    So the code length $n$ can be arbitrarily large.
  \end{proof}
  \begin{proof}[Proof of constant rate]
    There are $n = |V_{10}|+|V_{01}|$ qubits, 
      $m_z = |V_{00}|$ Z-stabilizers and $m_x = |V_{11}|$ X-stabilizers,
      so $k \ge n - m_z - m_x = |V_{10}|+|V_{01}|-|V_{00}|-|V_{11}|$.
    Because $X_\updownarrow, X_\leftrightarrow$ are regular bipartite graphs,
      we know the ratio $|V_{00}| : |V_{10}| : |V_{01}| : |V_{11}|
        = w_\uparrow w_\leftarrow : w_\downarrow w_\leftarrow
        : w_\uparrow w_\rightarrow : w_\downarrow w_\rightarrow$.
    Therefore, the rate, 
      $k/n \ge \frac{(w_\downarrow-w_\uparrow)(w_\leftarrow-w_\rightarrow)}
      {w_\downarrow w_\leftarrow + w_\uparrow w_\rightarrow} \ge r$.
  \end{proof}

  \begin{proof}[Proof of linear distance]
    By Lemma~\ref{lem:local-minimal-implies-linear-distance}, 
      to show linear distance
      it is sufficient to show 
      the chain complex $\cC$ and the dual chain complex $\cC^T$ have linear local minimal distance.
    By Corollary~\ref{cor:linear-local-minimal-distance}, 
      it is sufficient to check $\epsilon<1/12$ which holds by assumption.
  \end{proof}

  \begin{proof}[Proof of finite weight]
    From the code construction, we see
      each Z-stabilizer is connected to $w_\downarrow + w_\rightarrow$ qubits,
      each X-stabilizer is connected to $w_\uparrow + w_\leftarrow$ qubits,
      each qubit in $V_{10}$ is connected to $w_\uparrow + w_\rightarrow$ stabilizers,
      and each qubit in $V_{01}$ is connected to $w_\downarrow + w_\leftarrow$ stabilizers.
    So the code is LDPC with weight
      $\max(w_\downarrow + w_\rightarrow, w_\uparrow + w_\leftarrow, w_\uparrow + w_\rightarrow, w_\downarrow + w_\leftarrow) = \Theta(1)$.
  \end{proof}
\end{proof}
\section{Construct and prove linear time decoder for qLDPC} \label{sec:decoder}




In this section, we show that given the syndrome of a short error,
  there is a linear time decoder that corrects the error.
The decoder we use is the iterative greedy decoder
  which is a direct generalization of the decoder used in the lossless expander codes \cite{sipser1996expander}.
We describe the decoding algorithm in \ref{sec:decoding-algorithm},
  prove the correctness of the decoder in Sec.~\ref{sec:correctness-decoder},
  and show the decoder halts in linear time in Sec.~\ref{sec:linear-time-decoder}.

We first review some terminologies.
  The Z(X)-errors are the Z(X)-flips that act on the qubits.
  The goal of the decoder is to use the measurement outcome
    from the X(Z)-stabilizers
    to infer the Z(X)-errors.
  The measurement outcome is called the syndrome.
Using the language of the chain complex,
  a Z-error corresponds to a vector $c_1 \in \FF_2^{V_{01}} \oplus \FF_2^{V_{10}}$
  and the X-syndrome of the error corresponds to the vector $\partial c_1 \in \FF_2^{V_{11}}$.
Similarly, a X-error corresponds to $c'_1 \in \FF_2^{V_{01}} \oplus \FF_2^{V_{10}}$
  and the Z-syndrome of the error corresponds to $\partial^T c'_1 \in \FF_2^{V_{00}}$.

Without loss of generality, we focus on decoding the Z-error 
because our code construction is symmetric and  the code is a CSS code, i.e., we can decode the Z-error and the X-error separately.

\subsection{Decoding algorithm} \label{sec:decoding-algorithm}

Here, we describe a decoder that corrects the Z-error.

\begin{construction}[Decoding algorithm for the Z-error] 
  Input: $c_0 \in C_0$
  
  \begin{enumerate}
    \item Given the current syndrome, $v_{11} \subseteq V_{11}$, 
      find a vertex $x_{00} \in V_{00}$ 
      and a subset of its neighbors
        $n_{10}(x_{00}) \subseteq N_{V_{10}}(x_{00})$,
        $n_{01}(x_{00}) \subseteq N_{V_{01}}(x_{00})$
      such that after applying Z-flips at $n_{10}$, $n_{01}$,
      the number of syndromes that goes from 1 to 0 $\ge \beta \cdot$ the number of syndromes that is changed,
      where $\beta = 1 - 12 \epsilon$.
      (We call this the flippability condition and say $n_{10}(x_{00}), n_{01}(x_{00})$ is flippable.)
      Apply Z-flips at $n_{10}$, $n_{01}$
        and update the syndrome.
    \item Repeat, until no such vertex exists.
  \end{enumerate}
\end{construction}

When $\epsilon < 1/24$, we have $\beta > 1/2$, so the number of syndromes strictly reduces in each iteration.
This implies the number of iterations is at most the number of syndromes.


\subsection{Correctness of the decoder} \label{sec:correctness-decoder}

Now, we show the decoder correctly removes all the errors when the initial number of errors is small.

\begin{theorem} \label{thm:correctness-of-decoder}
  If the initial number of errors satisfies
  \begin{equation*}
    |c_1| < \min(w_\downarrow, w_\rightarrow) 
      (\frac{1}{2} - 6 \epsilon)
      (\min(\frac{\min (c_\leftrightarrow|V_{0*}|/w_\leftarrow, c_\updownarrow|V_{*0}|)}{w_\downarrow},
      \frac{\min (c_\updownarrow|V_{*0}|/w_\uparrow, c_\leftrightarrow|V_{0*}|)}{w_\rightarrow}) - 1)
      = \Theta(n),
  \end{equation*}
then the decoder removes all errors and the final codeword is the original closest codeword.
\end{theorem}

To show this, we show two lemmas: Lemma~\ref{lem:found-if-short} says, when the error is small, we can find a pair of flippable $(n_{10}(x_{00}), n_{01}(x_{00}))$. Lemma~\ref{lem:short-remain-short} says, if the initial error is small, then the error remains small throughout the algorithm. Note that because for quantum codes, error is only defined up to stabilizers, when we say the error of a syndrome, we pick the error with the smallest normalized weight. When the error has smallest normalized weight, it is guaranteed to be normalized locally minimal.

We denote $c_0(t)$ to be the syndrome at time $t$,
  and $c_1(t) = (v_{10}(t), v_{01}(t))$ to be the error with the smallest normalized weight $|c_1(t)|_w$
  such that $\partial c_1(t) = c_0(t)$.

\begin{lemma} [Found if short] \label{lem:found-if-short}
  If a small error has non empty syndrome, 
    $|v_{10}| < \min (c_\leftrightarrow |V_{0*}|/w_\leftarrow, c_\updownarrow |V_{*0}|), 
      |v_{01}| < \min (c_\updownarrow |V_{*0}|/w_\uparrow, c_\leftrightarrow |V_{0*}|)$.
    then one can find a vertex $x_{00}$ with $n_{10}(x_{00}), n_{01}(x_{00})$
    that is flippable.
\end{lemma}

\begin{lemma} [Short remains short] \label{lem:short-remain-short}
  If the initial number of normalized error is small,
    $|c_1(t=0)|_w < (\frac{1}{2} - 6 \epsilon)
      (\min(\frac{\min (c_\leftrightarrow|V_{0*}|/w_\leftarrow, c_\updownarrow|V_{*0}|)}{w_\downarrow},
      \frac{\min (c_\updownarrow|V_{*0}|/w_\uparrow, c_\leftrightarrow|V_{0*}|)}{w_\rightarrow}) - 1)$,
    then in all intermediate steps of the algorithm,
    $|v_{10}(t)| < \min (c_\leftrightarrow |V_{0*}|/w_\leftarrow, c_\updownarrow |V_{*0}|), 
      |v_{01}(t)| < \min (c_\updownarrow |V_{*0}|/w_\uparrow, c_\leftrightarrow |V_{0*}|)$.
\end{lemma}

Here, we prove Lemma~\ref{lem:short-remain-short} and Theorem~\ref{thm:correctness-of-decoder}.
And leave Lemma~\ref{lem:found-if-short} to Appendix~\ref{sec:proof-key-lemmas}.

\begin{proof}[Proof of Lemma~\ref{lem:short-remain-short}]
  The proof idea is to use the fact that the number of syndromes $|c_0(t)|$ strictly decreases,
    and use a bound between $|c_0(t)|_w$ and $|c_1(t)|_w$
    to show that even if $|c_1(t)|_w$ grows, it cannot grow much.

  On one hand,
    $|c_0| \le w_\rightarrow|v_{10}| + w_\downarrow|v_{01}|$,
    because $\deg_{11}(x_{10}) = w_\rightarrow$ and $\deg_{11}(x_{01}) = w_\downarrow$.
  So,
  \begin{equation}
    |c_0(t=0)|_w \le |c_1(t=0)|_w.
  \end{equation}
    
  On the other hand,
    by the small set LTC lemma~\ref{lem:small-set-LTC},
    if $|v_{10}(t)| < \min (c_\leftrightarrow |V_{0*}|/w_\leftarrow, c_\updownarrow |V_{*0}|), 
    |v_{01}(t)| < \min (c_\updownarrow |V_{*0}|/w_\uparrow, c_\leftrightarrow |V_{0*}|)$,
    then
  \begin{equation}
    (\frac{1}{2} - 6 \epsilon)|c_1(t)|_w \le |c_0(t)|_w.
  \end{equation}

  Suppose, $|v_{10}(t)| < \min (c_\leftrightarrow |V_{0*}|/w_\leftarrow, c_\updownarrow |V_{*0}|), 
    |v_{01}(t)| < \min (c_\updownarrow |V_{*0}|/w_\uparrow, c_\leftrightarrow |V_{0*}|)$ is violated at some time.
  Let $t=T$ be the first time it is violated.
  We have $|c_1(T)|_w = \min(\frac{|v_{10}(t)|}{w_\downarrow},
  \frac{|v_{01}(t)|}{w_\rightarrow})
  \ge \min(\frac{\min (c_\leftrightarrow |V_{0*}|/w_\leftarrow, c_\updownarrow |V_{*0}|)}{w_\downarrow},
  \frac{\min (c_\updownarrow |V_{*0}|/w_\uparrow, c_\leftrightarrow |V_{0*}|)}{w_\rightarrow})$.

  Because the flip in each iteration has normalized weight $\le 1$,
    $|c'_1(T)|_w - |c_1(T-1)|_w \le |c'_1(T)-c_1(T-1)|_w \le 1$,
    where $c'_1(T)$ is the direct application of the flip on $c_1(T-1)$.
  Because $c_1(T)$ has the smallest weight, we have $|c_1(T)|_w \le |c'_1(T)|_w$.
  Overall, $|c_1(T-1)|_w \ge |c'_1(T)|_w-1 \ge |c_1(T)|_w-1
    = \min(\frac{\min (c_\leftrightarrow |V_{0*}|/w_\leftarrow, c_\updownarrow |V_{*0}|)}{w_\downarrow},
    \frac{\min (c_\updownarrow |V_{*0}|/w_\uparrow, c_\leftrightarrow |V_{0*}|)}{w_\rightarrow}) - 1$.

  Now, because the initial the number of normalized error is assumed to satisfy
    $|c_1(t=0)|_w < (\frac{1}{2} - 6 \epsilon)
      (\min(\frac{\min (c_\leftrightarrow |V_{0*}|/w_\leftarrow, c_\updownarrow |V_{*0}|)}{w_\downarrow},
      \frac{\min (c_\updownarrow |V_{*0}|/w_\uparrow, c_\leftrightarrow |V_{0*}|)}{w_\rightarrow}) - 1)$,
    through 
    $(\frac{1}{2} - 6 \epsilon)|c_1(T-1)|_w \le |c_0(T-1)|_w \le |c_0(0)|_w \le |c_1(0)|_w$,
    we reach a contradiction,
    where we use the bound between $|c_0|_w, |c_1|_w$ for the first and the last inequalities,
      and $|c_0|_w$ strictly decreases for the second inequality.
  This shows $|v_{10}(t)| < \min (c_\leftrightarrow |V_{0*}|/w_\leftarrow, c_\updownarrow |V_{*0}|), 
    |v_{01}(t)| < \min (c_\updownarrow |V_{*0}|/w_\uparrow, c_\leftrightarrow |V_{0*}|)$
    hold for all $t>0$,
    and for $t=0$ one can show directly using $|v_{10}(0)| \le w_\downarrow |c_1(0)|_w$ and $|v_{01}(0)| \le w_\rightarrow |c_1(0)|_w$.
\end{proof}

\begin{proof}[Proof of Theorem~\ref{thm:correctness-of-decoder}]
  We first show the decoder removes all errors.
  Given the bound on the initial number of the (unnormalized) error $|c_1(t=0)|$,
    $|c_1(t=0)|_w \le |c_1(t=0)|/\min(w_\downarrow, w_\rightarrow)$
    satisfies the condition in the short remains short lemma~\ref{lem:short-remain-short}.
  Therefore, $|v_{10}(t)| < \min (c_\leftrightarrow |V_{0*}|/w_\leftarrow, c_\updownarrow |V_{*0}|), 
    |v_{01}(t)| < \min (c_\updownarrow |V_{*0}|/w_\uparrow, c_\leftrightarrow |V_{0*}|)$
    holds at all time.
  By the found if short lemma~\ref{lem:found-if-short},
    when the syndrome is non empty,
    the decoder will continue to reduce the syndrome.

  Now, we show the final codeword is the original closest codeword.
  Let $T$ be the time when the decoder halts.
  Because $|v_{10}(T)| < \min (c_\leftrightarrow |V_{0*}|/w_\leftarrow, c_\updownarrow |V_{*0}|), 
    |v_{01}(T)| < \min (c_\updownarrow |V_{*0}|/w_\uparrow, c_\leftrightarrow |V_{0*}|)$,
    by the small set LTC lemma~\ref{lem:small-set-LTC},
    $|c_0|_w = 0$ implies $|c_1|_w = 0$.
  Therefore, the decoder ends with the original codeword.
\end{proof}


\subsection{Running time of the decoder} \label{sec:linear-time-decoder}

In this section, we analyze the running time of the decoder. The naive implementation    where we scan through all possible flips,   results with quadratic time complexity.  So, we consider a refined version with additional preprocessing as follows.

\begin{enumerate}
  \item Preprocessing: 
    Given the syndromes,
      we make a list $Q$, 
      which contains all the candidate vertices that can be flipped by the decoder.
    Because there are $|V_{00}|=\Theta(n)$ many choices of $x_{00}$,
      and flippability can be determined in 
      $\Theta(w_\downarrow w_\rightarrow 2^{w_\downarrow+w_\rightarrow})=\Theta(1)$,
      where $2^{w_\downarrow+w_\rightarrow}$ is the number of possible choices for $n_{10}(x_{00}), n_{01}(x_{00})$
        (since $n_{10}(x_{00}) \subseteq N_{10}(x_{00})$ and $|N_{10}(x_{00})| = w_\downarrow$)
        and $\Theta(w_\downarrow w_\rightarrow)$ is the time to determine the flippable condition.
  \item Time complexity of each iteration:
    Take a vertex from the list $Q$,
      perform the flip,
      i.e. update the syndromes.
    After updating the syndromes,
      we update the list $Q$.
    It is enough to update the vertices neighbor of the updated syndromes.
    Because there are at most $w_\downarrow w_\rightarrow$ syndromes being updated,
      each syndrome neighbors to $w_\downarrow w_\rightarrow$ vertices in $V_{00}$.
    Combine with the time it takes to determine the flippability $\Theta(w_\downarrow w_\rightarrow 2^{w_\downarrow+w_\rightarrow})$,
      the time complexity for each iteration is $\Theta(w_\downarrow^3 w_\rightarrow^3 2^{w_\downarrow+w_\rightarrow})=\Theta(1)$.
  \item Number of iterations:
    When $\epsilon < \frac{1}{24}, \beta > \frac{1}{2}$,
      the number of syndromes strictly decreases in each iteration,
      so the number of iterations is at most the number of syndromes.
\end{enumerate}

Overall, the algorithm can be implemented to run in 
  $\Theta(w_\downarrow^3 w_\rightarrow^3 2^{w_\downarrow+w_\rightarrow}|c_0|)=\Theta(|c_0|)$.

We state result from the above discussion.
\begin{theorem}[Linear time decoder] \label{thm:linear-time-decoder}
  The decoder with preprocessing described above
    halts in time linear to the number of syndromes
    which is linear to the number of errors
    $\Theta(|c_0|) = O(|c_1|)$.
\end{theorem}

If the initial error is locally minimal, then $\Theta(|c_0|) = \Theta(|c_1|)$.
We write $O(|c_1|)$ to include the case
  where the initial error is not locally minimal.

\section{Conclusion} \label{sec:conclusion}
\subsection{Summary}

In this work, we construct good qLDPC codes using the balanced product \cite{breuckmann2021balanced} of two 2-sided lossless expander graphs. 
Assuming the existence of 2-sided lossless expander graphs with free group action \ref{con:2-sided-lossless-expander-with-symmetry},
  the resulting qLDPC codes have constant rate, linear distance and linear time decoders.

\subsection{Discussion}

Here, we compare our construction with two recent papers on similar topics \cite{panteleev2021asymptotically} \cite{dinur2021locally}. 
The common feature of these constructions is that all of them use the same kind of two-dimensional graph. 
The key difference is how one obtains the code from the graph.
Here, we illustrate 3 different ways to obtain a chain complex from the two-dimensional graph.

Take the notation in Definition~\ref{def:left-right-cayley-graph} and Figure~\ref{fig:left-right-cayley-graph}.
Denote vertices $V = V_{00} \cup V_{10} \cup V_{01} \cup V_{11}$,
horizontal edges $E^= = E_{0*} \cup E_{1*}$, 
vertical edges $E^{||} = E_{*0} \cup E_{*1}$,
and faces $F$.

Now, we can compare the chain complex obtained through 3 different methods.
In \cite{panteleev2021asymptotically},
  the chain complex is
  $\FF_2^{E^=} \rightarrow \FF_2^{F} \oplus \FF_2^{V} \rightarrow \FF_2^{E^{||}}$.
In \cite{dinur2021locally},
  the chain complex is
  $\FF_2^{F} \rightarrow \FF_2^{E^{||}} \oplus \FF_2^{E^=} \rightarrow \FF_2^{V}$.
In this work,
  the chain complex is
  $\FF_2^{V_{00}} \rightarrow \FF_2^{V_{10}} \oplus \FF_2^{V_{01}} \rightarrow \FF_2^{V_{11}}$.
For simplicity, we didn't include the base code in the examples above.

\subsection{Future work}

An important question left open in this paper
  is to construct 2-sided lossless expanders.
It is known that 1-sided lossless expanders with symmetry exist
  \cite{capalbo2002randomness},
  so one may hope to generalize their method
  and find a new version of zig-zag products
  that gives 2-sided lossless expanders.

\bibliographystyle{unsrt}
\bibliography{references.bib}

\appendix

\section{Lemma on lossless expanders} \label{sec:lossless-lemma}
In this section, we first show a simple lemma on unique expansion, 
  then show a bound of lossless expanders.
The bound from the Corollary~\ref{cor:lossless-split} puts enough constraint on lossless expanders
  that allows us to show the two key lemmas for 
  linear distance and the correctness of the linear time decoder in Appendix~\ref{sec:proof-key-lemmas}.

Comparing with Lemma 25 in \cite{lin2022c},
  the corollary~\ref{cor:lossless-split} implies Lemma 25
    and is a finer description
    that is needed for the proof of linear time decoder.
  For showing linear distance, it is enough to use the coarser bound in Lemma 25.

We first study a simple lemma.
Given a bipartite graph $(V_0, V_1, E)$.
If a vertex $x_1 \in V_1$ is neighbor to exactly one vertex in $v_0 \subseteq V_0$,
  we say $x_1$ is a unique neighbor of $v_0$.
We denote the unique neighbor of $v_0$ in $V_1$ as $N^{\textnormal{unique}}_{V_1}(v_0)$.

The first lemma says a lossless expander graph has many unique neighbors.
\begin{lemma} [Lossless expander implies unique expander] \label{lem:unique-expander}
  Let $(V_0, V_1, E)$ be a $(w_0, w_1)$-regular bipartite graph and $(c, \epsilon)$-lossless expander.
  Then for small set $v_0 \subseteq V_0$, $|v_0| < c |V_0|$,
  we have $|N^{\textnormal{unique}}_{V_1}(v_0)| \ge (1-2\epsilon) w_0 |v_0|$.
\end{lemma}

\begin{proof}
  Let $a_i$ be the number of $x_1$ with $\deg_{v_0}(x_1) = i$.
  
  By the definition of lossless expanders, $|N_{V_1}(v_0)| \ge (1-\epsilon) w_0 |v_0|$.
  So $\sum_{i=1}^{w_1} a_i \ge (1-\epsilon) w_0 |v_0|$.

  By counting the edges between $v_0$ and $N_{V_1}(v_0)$ in two ways,
  we have $w_0 |v_0| = \sum_{i=1}^{w_1} i a_i$.

  Together, $a_1 \ge 2 \sum_{i=1}^{w_1} a_i - \sum_{i=1}^{w_1} i a_i \ge (1-2\epsilon) w_0 |v_0|$,
  where we use $a_i \ge 0$ in the first inequality.
\end{proof}

Now, we turn to the main corollary. The goal of the following lemma is to extract a subgraph $Y$ for each small subgraph $X'$,
  such that $Y$ is a tree,
    and the remaining graph $Z$ has low degree at $V_0$.
  Roughly, it is saying the tree $Y$ is a good approximation of $X'$.

\begin{lemma} \label{lem:lossless-split}
  Given a $(w_0, w_1)$-regular bipartite graph, $(V_0, V_1, E)$,
    with 1-sided $(c, \epsilon)$-lossless expansion from $V_0$ to $V_1$.
  
  For each subgraph $X' = (v_0, v_1, e)$ with $|v_0| < c|V_0|$,
    we can partition $X'$ into two subgraphs $X' = Y \cup Z$,
    such that $\deg_{Y}(x_1) \le 1$
      for all $x_1 \in v_1$
    and $\deg_{Z}(x_0) \le \epsilon w_0$
      for all $x_0 \in v_0$.
\end{lemma}

We won't use the lemma directly. Instead, we will use its corollary.

Before stating the corollary, we recall some definitions.
A multiset is a modification of the concept of a set,
  where it is allowed for multiple instances.
For example, $\{a, a, b\}$ is a multiset
  with 2 instances of $a$ and 1 instance of $b$.
Another way to represent a multiset is to
  indicate the number of instances on the upper indices.
For example, $\{a, a, b\}$ can be written as $\{a^2, b\}$.

We say a multiset $A$ majorizes another multiset $B$
  if $\sum_{i=1}^k a_i \ge \sum_{i=1}^k b_i$
  for all $k = 1, ..., i_{\max}$
  where $i_{\max}=|A|=|B|$
    and $\{a_i\}_{i=1}^{i_{\max}}$ and $\{b_i\}_{i=1}^{i_{\max}}$
    are sorted sequences of $A$ and $B$ in the descending order.
This is denoted as $A \succeq B$.
In the case where $A$ and $B$ have different number of elements
  we append $0$s so that they become the same size
  and can be compared.

\begin{corollary} \label{cor:lossless-split}
  Given a $(w_0, w_1)$-regular bipartite graph, $(V_0, V_1, E)$,
    with 1-sided $(c, \epsilon)$-lossless expansion from $V_0$ to $V_1$.

  For any small subset $v_1 \subseteq V_1$ with $|v_1| < c|V_0|/w_1$,
    there exists $n(x_0) \subseteq N_{v_1}(x_0)$ for each $x_0 \in V_0$,
    such that $\{ n(x_0) : x_0 \in V_0 \}$ is a partition of $v_1$
      and $|N_{v_1}(x_0) - n(x_0)| \le \epsilon w_0$.
  Furthermore, $\{|N_{v_1}(x_0) - n(x_0)|\}_{x_0 \in V_0} \preceq \{\epsilon w_0^{\lceil \frac{w_1}{\epsilon w_0} |v_1| \rceil}\}$.
\end{corollary}

\begin{proof}
  Consider $X' = (N_0(v_1), v_1, N_E(v_1))$
    where $N_0(v_1)$ is set of neighboring vertices of $v_1$ in $V_0$,
    and $N_E(v_1)$ is set of neighboring edges of $v_1$ in $E$.
  Because $|v_1| < c|V_0|/w_1$, we have $|N_0(v_1)| < c|V_0|$,
    so Lemma~\ref{lem:lossless-split} applies.
  Say the partition is $X' = Y \cup Z$.

  Because $\deg_{X'}(x_1) \ge 1$ and $\deg_{Y}(x_1) \le 1$ for all $x_1 \in v_1$,
    we can find a subgraph $Y'$ such that 
      $Y \subseteq Y' \subseteq X'$
      with $\deg_{Y'}(x_1) = 1$ for all $x_1 \in v_1$,
      by adding an edge to $x_1$ when $\deg_{Y}(x_1) = 0$.
  Now, by setting $n(x_0) = N_{Y'}(x_0)$,
    we have $\{ n(x_0) : x_0 \in V_0 \}$ is a partition of $v_1$
    and $|N(x_0) - n(x_0)| \le \epsilon w_0$.
  The final majorization inequality holds,
    because $|N(x_0) - n(x_0)| \le \epsilon w_0$ and the number of edges emanating from $v_1$ is $\le w_1 |v_1|$.
\end{proof}

Now, we go back to the lemma.
The lemma is similar to Hall's marriage theorem
  and can be shown similarly using max-flow min-cut theorem
  together with a lemma that bounds the number of edges from the number of vertices.
We first show the bound between edges and vertices, 
  then we review the max-flow min-cut theorem.

\begin{lemma}
  Given a $(w_0, w_1)$-regular bipartite graph, $(V_0, V_1, E)$,
    with 1-sided $(c, \epsilon)$-lossless expansion from $V_0$ to $V_1$.
  For any subsets $v_0 \subseteq V_0$ and $v_1 \subseteq V_1$,
    with $|v_0| < c|V_0|$,
    the number of edges between $v_0$ and $v_1$, $|E(v_0, v_1)|$,
    is bounded by
  \begin{equation} \label{eq:ineq-lossless-v1}
    |E(v_0, v_1)| = \sum_{x_0 \in v_0} \deg_{v_1}(x_0) \le \epsilon w_0 |v_0| + |v_1|,
  \end{equation}

  Moreover,
  \begin{equation} \label{eq:ineq-lossless-v2}
    \sum_{x_0 \in v_0} \max(\deg_{v_1}(x_0) - \epsilon w_0, 0) \le |v_1|.
  \end{equation}
\end{lemma}

\begin{proof}
  We prove the first inequality by 
    consider the graph consist of
      $v_0$ and its neighbors $N_{V_1}(v_0)$,
    then remove vertices $N_{V_1}(v_0)-v_1$ and the connected edges.

  First, the graph form by $v_0$ and $N_{V_1}(v_0)$ has $w_0 |v_0|$ edges.
  Next, when we remove $N_{V_1}(v_0)-v_1$ vertices, 
    we remove at least $|N_{V_1}(v_0)|-|v_1|$ edges.
  So 
  \begin{equation*}
    |E(v_0, v_1)| \le w_0 |v_0| - (|N_{V_1}(v_0)|-|v_1|) \le \epsilon w_0 |v_0| + |v_1|,
  \end{equation*}
    where the last inequality follows from the lossless assumption.
  
  The second inequality can be derived from the first inequality
    by writing
    \begin{equation*}
      \sum_{x_0 \in v_0} \max(\deg_{v_1}(x_0) - \epsilon w_0, 0)
      = \sum_{x_0 \in \{x_0 \in v_0 : \deg_{v_1}(x_0) > \epsilon w_0\} } 
        \deg_{v_1}(x_0) - \epsilon w_0
      \le |v_1|.
    \end{equation*}
\end{proof}

Here, we review the definition of flow, cut, and the max-flow min-cut theorem.

\begin{definition}[Flow network, flow and cut]
  Let $N = (V, E, s, t, c)$ be a flow network
    where $(V, E)$ forms a directed graph,
    $s, t \in V$ are the source and the sink,
    and $c: E \rightarrow \RR$ is the capacity function.
  
  A flow is a function $f: E \rightarrow \RR$ that satisfies
  \begin{enumerate}
    \item Capacity constraint: For every edge $(u, v) \in E$,
      $f(u,v) \le c(u,v)$.
    \item Conservation of flows: For each vertex $v \in V$ besides $s$ and $t$,
      $\sum_{\{u:(u,v) \in E\}} f(u,v) = \sum_{\{w:(v,w) \in E\}} f(v,w)$.
  \end{enumerate}

  The value of a flow is defined by
  \begin{equation}
    |f| = \sum_{\{v:(s,v) \in E\}} f(s,v) = \sum_{\{v:(v,t) \in E\}} f(v,t),
  \end{equation}
    where the last equality follows from flow conservation.

  A cut $(S,T)$ is a partition of V such that
    $s \in S$ and $t \in T$.
  
  The capacity of a cut is the sum of the capacities at the boundary of $S$ and $T$
  \begin{equation}
    c(S,T) = \sum_{(u,v) \in E, u \in S, v \in T} c(u,v).
  \end{equation}
  
\end{definition}

The maximum flow problem is to maximize $|f|$.
The minimum cut problem is to minimize $c(S, T)$.

\begin{theorem}[Max-flow min-cut] \label{thm:max-flow-min-cut}
  The maximum value of a flow is equal to the minimum capacity of a cut.

  Furthermore, when the capacities in a flow network are integers,
    there is a maximum flow such that the flow on each edge is an integer.
\end{theorem}

Finally, we are ready to prove the main lossless lemma.

\begin{proof}[Proof of Lemma~\ref{lem:lossless-split}]
  Recall that $X' = (V'_0, V'_1, E')$ is the subgraph of interest
    and our goal is to find a partition into two subgraphs $X' = Y \cup Z$.

  Consider the flow network $(V, E, s, t, c)$ where
  \begin{equation}
    V = \{s\} \cup V'_0 \cup V'_1 \cup \{t\},
  \end{equation}

  \begin{equation}
    E = \{(s, x_0) : x_0 \in V'_0\} \cup e \cup \{(x_1, t) : x_1 \in V'_1\},
  \end{equation}

  \begin{equation}
    c(s, x_0) = \max(\deg_{V'_1}(x_0)-\epsilon w_0, 0),
  \end{equation}

  \begin{equation}
    c(x_0, x_1) = 1,
  \end{equation}

  \begin{equation}
    c(x_1, t) = 1.
  \end{equation}

  \begin{figure}
    \centering
    \begin{tikzpicture}
      \newcommand{\w}{1.5}

\fill (0,0) circle[radius=1pt] node (s) [label=below:$s$] {};

\fill (\w, 1) circle[radius=1pt] node (x0) {};
\fill (\w, 0) circle[radius=1pt] node (y0) {};
\fill (\w, -1) circle[radius=1pt] node (z0) [label=below:$V'_0$] {};

\fill (2*\w, 1) circle[radius=1pt] node (x1) {};
\fill (2*\w, 0) circle[radius=1pt] node (y1) {};
\fill (2*\w, -1) circle[radius=1pt] node (z1) [label=below:$V'_1$] {};

\fill (3*\w,0) circle[radius=1pt] node (t) [label=below:$t$] {};

\draw [->] (s) -- (x0);
\draw [->] (s) -- (y0);
\draw [->] (s) -- (z0);

\draw [->] (x0) -- (x1);
\draw [->] (x0) -- (y1);

\draw [->] (y0) -- (x1);
\draw [->] (y0) -- (y1);
\draw [->] (y0) -- (z1);

\draw [->] (z0) -- (y1);
\draw [->] (z0) -- node[auto, swap] {$E'$} (z1);

\draw [->] (x1) -- (t);
\draw [->] (y1) -- (t);
\draw [->] (z1) -- (t);

\draw (0.5*\w, -2.5) node [rotate=-90]
  {$\left.\rule{0pt}{\dimexpr\w cm/2}\right\}$};
\draw (.5*\w, -3.5) node [overlay] {$c=\max(\deg_{V'_1}(x_0)-\epsilon w_0, 0)$};
\draw (1.5*\w, -2.5) node {$c=1$};
\draw (2.5*\w, -2.5) node {$c=1$};
    \end{tikzpicture}
    \caption{The flow network for proving Lemma~\ref{lem:lossless-split}.}
    \label{fig:flow-network}
  \end{figure}

  Here, we claim the existence of an integer flow which proves the theorem
  and we prove the claim afterwards.
  \begin{claim}
    There exist an integer flow $f$ with value 
    $\sum_{x_0 \in V'_0} \max(\deg_{X'}(x_0) - \epsilon w_0, 0)$.
  \end{claim}

  Assuming the exists of the flow $f$,
    we define $Y = (V'_0, V'_1, \{e \in E': f(e) = 1\})$
    to be the subgraph consists of edges with flows equal to 1
    and define $Z = X' - Y$ to be its complement.

  We show the subgraphs satisfy the condition.
  We have $\deg_{Y}(x_1) \le 1$, because of capacity $c(x_1, t) = 1$ and flow conservation.
  Now, we show $\deg_{Z}(x_0) \le \epsilon w_0$.
  Because the cut $S = \{s\}$, $T = V'_0 \cup V'_1 \cup \{t\}$
    has capacity $\sum_{x_0 \in V'_0} \max(\deg_{X'}(x_0) - \epsilon w_0, 0)$,
    $f$ is saturated on this cut,
    so $f(s, x_0) = \max(\deg_{X'}(x_0) - \epsilon w_0, 0)$.
  By flow conservation
    $\deg_{Y}(x_0) = f(s, x_0) = \max(\deg_{X'}(x_0) - \epsilon w_0, 0)$
    which implies $\deg_{Z}(x_0) \le \epsilon w_0$.

  Now we suffice to prove the claim.

  \begin{proof} [Proof of the claim]
    By max-flow min-cut theorem \ref{thm:max-flow-min-cut}, 
      it is sufficient to show all cuts have value greater than 
      $\sum_{x_0 \in V'_0} \max(\deg_{V'_1}(x_0) - \epsilon w_0, 0)$.
    We will show this using Equation~\ref{eq:ineq-lossless-v2} in Lemma~\ref{lem:lossless-split}.
    
    Given a cut $S = \{s\} \cup S_0 \cup S_1$, $T = T_0 \cup T_1 \cup \{t\}$,
      where $S_0 \cup T_0 = V'_0$, $S_1 \cup T_1 = V'_1$.
    The capacity of the cut is 
    \begin{eqs}
      c(S, T) &= c(s, T_0) + c(S_0, T_1) + c(S_1, t) \\
              &= \sum_{x_0 \in T_0} \max(\deg_{V'_1}(x_0) - \epsilon w_0, 0)
              + E(S_0, T_1) + |S_1| \\
              &= \sum_{x_0 \in T_0} \max(\deg_{V'_1}(x_0) - \epsilon w_0, 0)
              + \sum_{x_0 \in S_0} \deg_{T_1}(x_0) + |S_1|.
    \end{eqs}
    
    Now, we show the value is $\ge \sum_{x_0 \in V'_0} \max(\deg_{V'_1}(x_0) - \epsilon w_0, 0)$.
    \begin{eqs}
      & \sum_{x_0 \in V'_0} \max(\deg_{V'_1}(x_0) - \epsilon w_0, 0) \le c(S, T) \\
      \Leftrightarrow &\sum_{x_0 \in S_0} \max(\deg_{S_1}(x_0) + \deg_{T_1}(x_0) - \epsilon w_0, 0)
        \le \sum_{x_0 \in S_0} \deg_{T_1}(x_0) + |S_1| \\
      \Leftarrow &\sum_{x_0 \in S_0} \max(\deg_{S_1}(x_0) - \epsilon w_0, 0) \le |S_1|,
    \end{eqs}
    where the last inequality follows from Equation~\ref{eq:ineq-lossless-v2}
      which uses $|V'_0| \le c |V_0|$.
  
    Therefore, all cuts have capacities
      $\ge \sum_{x_0 \in V'_0} \max(\deg_{V'_1}(x_0) - \epsilon w_0, 0)$,
      which together with the max-flow min-cut theorem~\ref{thm:max-flow-min-cut} implies the existence of the integer flow.
  \end{proof}
\end{proof}
\section{Chain complex from balanced product}

In this section, we first provide more detail on the construction of the chain complex, then we prove the chain complex is well defined.

To get the code, we need to obtain vector spaces and linear maps from the graph. We do so by taking the adjacency matrix. 
Recall the adjacency matrix of a bipartite graph
  is a linear map
  $L((V_0, V_1, E)): \FF_2^{V_0} \rightarrow \FF_2^{V_1}$.

Before going further, we introduce a convenient notation.
Because of the one-one correspondence between
  the subsets of $V$ and vectors in $\FF_2^{V}$
  by the map $\v \mapsto \sum_{x \in \v} e_x$
    where $\v \subseteq V$,
  we abuse the notation by overloading both use cases.
For example, we write both $x \in \v$ and $\v \in \FF_2^{V}$,
  where the first $\v$ is interpreted as a set with $x$ as its element
  and the second $\v$ is interpreted as a vector $\v \in \FF_2^{V}$.
Similarly, $x_0$ can be interpreted both as an element in $V_0$
  and the basis vector $e_{x_0}$.
  
Now, we generalize the construction of linear maps from bipartite graphs to
  balanced product of bipartite graphs.
Given a balanced product graph $(V^*, E^*, F)$,
  we obtain the vector spaces 
  $\FF_2^{V_{00}}, \FF_2^{V_{10}}, \FF_2^{V_{01}}, \FF_2^{V_{11}}$
  and the linear maps
  \begin{itemize}
    \item $L_{00 \rightarrow 10}=L(E_{*0}): \FF_2^{V_{00}} \rightarrow \FF_2^{V_{10}}$,
    \item $L_{00 \rightarrow 01}=L(E_{0*}): \FF_2^{V_{00}} \rightarrow \FF_2^{V_{01}}$,
    \item $L_{10 \rightarrow 11}=L(E_{1*}): \FF_2^{V_{10}} \rightarrow \FF_2^{V_{11}}$,
    \item $L_{01 \rightarrow 11}=L(E_{*1}): \FF_2^{V_{01}} \rightarrow \FF_2^{V_{11}}$,
  \end{itemize}
  where
  $L_{00 \rightarrow 10}(z_{00}) = \sum_{z_{10} \in N_{10}(z_{00})} z_{10}$,
  $L_{00 \rightarrow 01}(z_{00}) = \sum_{z_{01} \in N_{01}(z_{00})} z_{01}$,
  $L_{10 \rightarrow 11}(z_{10}) = \sum_{z_{11} \in N_{11}(z_{10})} z_{11}$,
  $L_{01 \rightarrow 11}(z_{01}) = \sum_{z_{11} \in N_{11}(z_{01})} z_{11}$.

We claim that these 4 linear maps form a chain complex.
  $\FF_2^{V_{00}}
    \xrightarrow{\partial_2} \FF_2^{V_{10}} \oplus \FF_2^{V_{01}}
    \xrightarrow{\partial_1} \FF_2^{V_{11}}$,
  where
  \begin{equation}
    \partial_2(v_{00}) = (L_{00 \rightarrow 10}(v_{00}), L_{00 \rightarrow 01}(v_{00})),
  \end{equation}
  \begin{equation}
    \partial_1((v_{10}, v_{01})) = L_{10 \rightarrow 11}(v_{10}) + L_{01 \rightarrow 11}(v_{01}).
  \end{equation}
For the chain complex to be well defined,
  we need to show the condition $\partial_1 \partial_2 = 0$,
  which is same as showing 
  $L_{10 \rightarrow 11} L_{00 \rightarrow 10}(v) + L_{01 \rightarrow 11} L_{00 \rightarrow 01}(v) = 0$.

To prove it, we first show a lemma.

\begin{lemma} [Square completion lemma] \label{lem:square-completion}
  Given $z_{00} \in V_{00}, z_{10} \in V_{10}, z_{01} \in V_{01}$
  and $(z_{00}, z_{10}) \in E_{*0}$, $(z_{00}, z_{01}) \in E_{0*}$,
  then there exists a unique vertex $z_{11} \in V_{11}$ that completes the square $(z_{00}, z_{10}, z_{01}, z_{11}) \in F$.
\end{lemma}

\begin{proof} [Proof of square completion lemma]
  By the definition of balanced product, we can write
    $z_{00} = [(x_0, y_0)], z_{10} = [(x^a_1, y^a_0)]], z_{01} = [(x^b_0, y^b_1)]$,
    and because of the group action is free there exist unique $g^a, g^b \in G$
    such that $y^a_0 = g^a y^a_0, x^a_0 = g^b x^b_0$.
    Let $x_1 = g^a x^a_1, y_1 = g^b y^b_1$.

  We first show existence.
    Consider the face $[((x_0, y_0), (x_1, y_0), (x_0, y_1), (x_1, y_1))] \in F$, 
      we see $z_{11} = [(x_1, y_1)]$ satisfies the condition.

  Now, we show uniqueness.
    Say $(z_{00}, z_{10}, z_{01}, z'_{11}) \in F$.
    By the definition of balanced product, we can write 
      $(z_{00}, z_{10}, z_{01}, z_{11}) = [((x^c_0, y^c_0), (x^c_1, y^c_0), (x^c_0, y^c_1), (x^c_1, y^c_1))]$,
    and there exists unique $g^c \in G$
      such that $x_0 = g^c x^c_0, y_0 = g^c y^c_0$.
      Let $x'_1 = g^c x^c_1, y'_1 = g^c y^c_1$.
    Because $(x_1, y_0) \sim (x^a_1, y^a_0) \sim (x^c_1, y^c_0) \sim (x'_1, y_0)$
      and because the group action is free,
      we have $x'_1 = x_1$.
    Similarly $y'_1 = y_1$.
    So $z'_{11} = [(x^c_1, y^c_1)] = [(x'_1, y'_1)] = [(x_1, y_1)] = z_{11}$.
\end{proof}

For this unique vertex, we denote $z_{11} = z_{10} \times_{z_{00}} z_{01}$.
Using a similar argument, we can also define
  $z_{10} = z_{11} \times_{z_{01}} z_{00}$,
  $z_{01} = z_{00} \times_{z_{10}} z_{11}$,
  $z_{00} = z_{01} \times_{z_{11}} z_{10}$.

Now, we show the condition
  $L_{10 \rightarrow 11} L_{00 \rightarrow 10}(v) + L_{01 \rightarrow 11} L_{00 \rightarrow 01}(v) = 0$.
By linearity, it is sufficient to show 
  $L_{10 \rightarrow 11} L_{00 \rightarrow 10}(z_{00}) + L_{00 \rightarrow 01} L_{01 \rightarrow 11}(z_{00}) = 0$ for each $z_{00} \in V_{00}$.

If we expand the summations,
  $L_{10 \rightarrow 11} L_{00 \rightarrow 10}(z_{00}) = \sum_{z_{11} \in N_{11}(z_{10})} \sum_{z_{10} \in N_{10}(z_{00})} z_{11}$,
  we see the number of $z_{11}$ appearing in $L_{10 \rightarrow 11} L_{00 \rightarrow 10}(z_{00})$ is equal to
    $|\{z_{10}: z_{10} \in N_{10}(z_{00}) \cap N_{10}(z_{11})\}|$.
Similarly, the number of $z_{11}$ appearing in $L_{01 \rightarrow 11} L_{00 \rightarrow 01}(z_{00})$ is equal to
$|\{z_{01}: z_{01} \in N_{01}(z_{00}) \cap N_{01}(z_{11})\}|$.

Now, we show a bijection between $\{z_{10}: z_{10} \in N_{10}(z_{00}) \cap N_{10}(z_{11})\}$ and $\{z_{01}: z_{01} \in N_{01}(z_{00}) \cap N_{01}(z_{11})\}$.
This would implies the total number of $z_{11}$ appearing in $L_{10 \rightarrow 11} L_{00 \rightarrow 10}(z_{00}) + L_{00 \rightarrow 01} L_{01 \rightarrow 11}(z_{00})$ is even
which proves the condition $\partial_1 \partial_2 = 0$.

To show the bijection, we use the square completion lemma~\ref{lem:square-completion}.
  Define $f: \{z_{10}: z_{10} \in N_{10}(z_{00}) \cap N_{10}(z_{11})\} \rightarrow \{z_{01}: z_{01} \in N_{01}(z_{00}) \cap N_{01}(z_{11})\}$ to be
    $f(z_{10}) = z_{00} \times_{z_{10}} z_{11}$
  and $g: \{z_{01}: z_{01} \in N_{01}(z_{00}) \cap N_{01}(z_{11})\} \rightarrow \{z_{10}: z_{10} \in N_{10}(z_{00}) \cap N_{10}(z_{11})\}$ to be 
    $g(z_{01}) = z_{11} \times_{z_{01}} z_{00}$.
  It is clear that $f(g(z_{10})) = z_{10}$ and $g(f(z_{01})) = z_{01}$.
\section{Lemmas for Theorem~\ref{thm:good-qLDPC}} \label{sec:distance-lemma}
In this section, we provide relevant materials for the proof of good qLDPC code, Theorem~\ref{thm:good-qLDPC}.

We first prove a simple lemma that shows the one-dimensional subgraph of the balanced product graph remains a lossless expander.
Then we introduce locally minimal distance and show that small set LTC lemma implies linear locally minimal distance, which further implies linear distance.

\subsection{Lossless expander}

\begin{lemma} \label{lem:freely-implies-copy}
Given two graphs $X_{\updownarrow} = (V_{0*}, V_{1*}, E_{\updownarrow})$ and $X_{\leftrightarrow} = (V_{*0}, V_{*1}, E_{\leftrightarrow})$ with free $G$-invariant action. Then for the one-dimensional subgraph of $X_{\updownarrow} \times_G X_{\leftrightarrow}$,
$(V_{00}, V_{10}, E_{*0})$ is isomorphic to $|V_{*0}/G|$ copies of $X_{\updownarrow}$. Similarly,  
$(V_{01}, V_{11}, E_{*1})$ is isomorphic to $|V_{*1}/G|$ copies of $X_{\updownarrow}$,
$(V_{00}, V_{01}, E_{0*})$ is isomorphic to $|V_{0*}/G|$ copies of $X_{\leftrightarrow}$, and 
$(V_{10}, V_{11}, E_{1*})$ is isomorphic to $|V_{1*}/G|$ copies of $X_{\leftrightarrow}$.
\end{lemma}

\begin{proof}
  Here, we show the case for $(V_{00}, V_{10}, E_{*0})$.
    Other cases follow similarly.
  Before quotienting,
    the 1d subgraph $(V'_{00}, V'_{10}, E'_{*0})$ of
    the hypergraph product $X_{\updownarrow} \times X_{\leftrightarrow}$
    is isomorphic to $|V_{*0}|$ copies of $X_{\updownarrow}$
    each labeled by $V_{*0}$.
  After quotienting,
    the copies in the same orbit of $V_{*0}$
    are identified into 1 copy of $X_{\updownarrow}$.
    So we are left with $|V_{*0}/G|$ copies of $X_{\updownarrow}$.
\end{proof}

\begin{corollary} \label{cor:freely-implies-lossless}
  Under the same assumption in \ref{lem:freely-implies-copy}.
  Futhermore, assume $X_{\updownarrow}$ is a 1-sided $(c_\updownarrow, \epsilon_\updownarrow)$-lossless expander, $X_{\leftrightarrow}$ is a 1-sided $(c_\leftrightarrow, \epsilon_\leftrightarrow)$-lossless expander.
  Then $(V_{00}, V_{10}, E_{*0})$ is a 1-sided $(c_\updownarrow |V_{0*}|/|V_{00}|, \epsilon_\updownarrow)$-lossless expander, 
  $(V_{01}, V_{11}, E_{*1})$ is a 1-sided $(c_\updownarrow |V_{0*}|/|V_{01}|, \epsilon_\updownarrow)$-lossless expander,
  $(V_{00}, V_{01}, E_{0*})$ is a 1-sided $(c_\leftrightarrow |V_{*0}|/|V_{00}|, \epsilon_\leftrightarrow)$-lossless expander,
  $(V_{10}, V_{11}, E_{1*})$ is a 1-sided $(c_\leftrightarrow |V_{*0}|/|V_{10}|, \epsilon_\leftrightarrow)$-lossless expander.
\end{corollary}

\begin{proof}
  Here, we show the case for $(V_{00}, V_{10}, E_{*0})$.
    Other cases follow similarly.
  By definition, any small set $v_{0*} \subseteq V_{0*}$, $|v_{0*}| < c_\updownarrow |V_{0*}|$,
    satisfies $|N_{V_{1*}}(v_{0*})| \ge (1-\epsilon_\updownarrow) w_\downarrow |v_{0*}|$.

  Now, given a small set $v_{00} \subseteq V_{00}$, $|v_{00}| < c_\updownarrow |V_{0*}|$.
    Let $v_{00} = \cup_{i=1}^{|V_{*0}/G|} v_{00, i}$, where $v_{00, i}$ is the intersection of $v_{00}$ with the $i$-th copy of $X_{\updownarrow}$.
  Because each $v_{00, i}$ is small, $|v_{00, i}| < c_\updownarrow |V_{0*}|$,
    the size of its neighbor $|N_{V_{10}}(v_{00, i})|$ has size at least $(1-\epsilon_\updownarrow)|v_{00, i}|$
  Therefore, $|N_{V_{10}}(v_{00})| = \sum_{i=1}^{|V_{*0}/G|} |N_{V_{10}}(v_{00, i})|
    \ge \sum_{i=1}^{|V_{*0}/G|} (1-\epsilon_\updownarrow)|v_{00, i}| = (1-\epsilon_\updownarrow)|v_{00}|$.
  This implies $(V_{00}, V_{10}, E_{*0})$ is a 1-sided $(c_\updownarrow |V_{0*}|/|V_{00}|, \epsilon_\updownarrow)$-lossless expander.
\end{proof}

\subsection{Locally minimal} \label{sec:local-minimal}

In this section, we introduce a variant of the locally minimal distance
  \cite{evra2020decodable}
  \cite{kaufman2014ramanujan}
  \cite{panteleev2021asymptotically},
  the normalized locally minimal distance
  and show that the normalized locally minimal distance is a lower bound of
  the distance.

We first review the definition of local minimality.

\begin{definition}[Locally minimal]
  Given a chain complex 
    $C_{i+1} \xrightarrow{\partial_{i+1}} C_{i}$.
  A vector $c_i \in C_i$ is locally minimal if 
    for any basis vector $e_{i+1} \in C_{i+1}$
    \begin{equation}
      |c_i + \partial_{i+1} e_{i+1}| \ge |c_i|.
    \end{equation}
\end{definition}

The definition of locally minimal is related to the greedy flipping decoder of the expander code.

\begin{definition}[Greedy flipping algorithm]
  Input: $c_i \in C_i$.
  \begin{enumerate}
    \item If there exists a basis vector $e_{i+1} \in C_{i+1}$,
      such that $|c_i + \partial_{i+1} e_{i+1}| < |c_i|$,
      replace $c_i$ with $c_i + \partial_{i+1} e_{i+1}$. 
    \item Repeat until no such $e_{i+1}$ exists. Output $c_i$.
  \end{enumerate}
\end{definition}

Any output of a greedy flipping algorithm is locally minimal.
Note that $c_i$ strictly decreases in each iteration.
So the algorithm halts in $|c_i|$ steps.
Note that $\partial_{i+1} c_i$ does not change throughout the algorithm
  because in each iteration the change is $0$, $\partial (c_i + \partial e_{i+1}) - \partial c_i = \partial \partial e_i = 0$.
We refer the process of 
  replacing $c_i$ with $c_i + \partial_{i+1} e_{i+1}$ flipping, 
  because in $\FF_2$, the bits flip between 0 and 1.

In our context, we consider a variant, the normalized locally minimal,
  where we normalize the weight
    before comparing $c_i + \partial_{i+1} e_{i+1}$ and $c_i$.
  The purpose of performing this additional normalization is to make the statement more natural.

The normalization is determined through the following discussion.
For a chain complex constructed from balanced product of regular bipartite graphs,
  $\partial e_2$ flips 
  $w_\downarrow$ bits in $\FF_2^{V_{10}}$
  and $w_\rightarrow$ bits in $\FF_2^{V_{01}}$.
So we will weight the components in $\FF_2^{V_{10}}$ with $1/w_\downarrow$,
  and the components in $\FF_2^{V_{01}}$ with $1/w_\rightarrow$.

\begin{definition}[Normalized locally minimal]
  Given a chain complex 
    $C_2 \xrightarrow{\partial_2} C_1 
    \xrightarrow{\partial_1} C_0$
    constructed from balanced product of regular bipartite graphs.
  A vector $c_1 = (v_{10}, v_{01}) \in \Ker \partial_1$ is normalized locally minimal if 
    for any basis vector $e_2 \in C_2$
    \begin{equation}
      |c_1 + \partial e_2|_w \ge |c_1|_w,
    \end{equation}
    where $|(v_{10}, v_{01}))|_w = |v_{10}|/w_\downarrow + |v_{01}|/w_\rightarrow$.
\end{definition}

The greedy flipping algorithm still applies.
Because in each step, $|c_1|_w$ strictly decreases by at least $1/\max(w_\downarrow, w_\rightarrow)$,
  the algorithm halts in $|c_1|_w \max(w_\downarrow, w_\rightarrow) \le |c_1| \max(w_\downarrow, w_\rightarrow)/\min(w_\downarrow, w_\rightarrow) = \Theta(|c_1|)$ steps.

From now on, we only consider the chain complex
  constructed from balanced product of regular bipartite graphs,
  and locally minimal always means normalized locally minimal.

Now, we define the locally minimal distance.

\begin{definition}[Locally minimal distance]
  Given a chain complex 
    $\cC: C_2 \xrightarrow{\partial_2} C_1 
    \xrightarrow{\partial_1} C_0$.
  The locally minimal distance $d_1^{LM}(\cC)$ is 
    the minimal weight of all the non trivial locally minimal vectors.
  Formally,
  \begin{equation}
    d_1^{LM}(\cC) = \min_{c_1 \in \Ker (\partial_1), c_1\text{ is locally minimal}, c_1 \ne 0} |c_1|.
  \end{equation}
\end{definition}

Finally, we show the locally minimal distance is a lower bound of
  the distance.

\begin{lemma}[Linear locally minimal distance implies linear distance]
  \label{lem:local-minimal-implies-linear-distance}
  Given a chain complex $\cC$.
  Then 
    \begin{equation}
      d_1(\cC) \ge d_1^{LM}(\cC).
    \end{equation}
\end{lemma}

\begin{proof}
  Recall $d_1(\cC) = \min_{c_1 \in \Ker \partial_1 - \Ima \partial_2} |c_1|$.
  So $d_1(\cC) = |c_1|$ for some $c_1 \in \Ker \partial_1 - \Ima \partial_2$.
  We show that such $c_1$ is locally minimal.

  Because $c_1 \notin \Ima \partial_2$, 
    we have $c_1 + \partial e_2 \notin \Ima \partial_2$.
  Because $c_1$ has the smallest weight in $\Ker \partial_1 - \Ima \partial_2$,
    we have $|c_1| \le |c_1 + \partial e_2|$.
  Therefore, $c_1$ is locally minimal 
    and $d_1(\cC) \ge d_1^{LM}(\cC)$.
\end{proof}

\begin{corollary}[Linear locally minimal distance] \label{cor:linear-local-minimal-distance}
  Under the same assumption as in the lemma~\ref{lem:small-set-LTC}
    and $\epsilon < 1/12$,
    we have
  \begin{equation}
    d_1^{LM}(\cC) \ge \min (c_\leftrightarrow|V_{0*}|/w_\leftarrow, c_\updownarrow|V_{*0}|, 
      c_\updownarrow|V_{*0}|/w_\uparrow, c_\leftrightarrow|V_{0*}|).
  \end{equation}

  When $|V_{0*}| = \Theta(|G|), |V_{*0}| = \Theta(|G|)$,
    and $w_\downarrow$, $w_\uparrow$, $w_\rightarrow$, $w_\leftarrow$, $c$ are $\Theta(1)$,
    we have $\min (c_\leftrightarrow|V_{0*}|/w_\leftarrow, c_\updownarrow|V_{*0}|, 
    c_\updownarrow|V_{*0}|/w_\uparrow, c_\leftrightarrow|V_{0*}|) = \Theta(n)$,
    where $n = |V_{10}| + |V_{01}|$.
  So,
  \begin{equation}
    d_1^{LM}(\cC) \ge \Theta(n),
  \end{equation}
  which means the locally minimal distance is linear.
\end{corollary}

\begin{proof}
  Recall the definition of locally minimal distance,
    $d_1^{LM}(\cC) = \min_{c_1 \in \Ker \partial_1, c_1\text{ is locally minimal}, c_1 \ne 0} |c_1|$.

  From lemma~\ref{lem:small-set-LTC} we know if 
    $|v_{10}| < \min(c_\leftrightarrow|V_{0*}|/w_\leftarrow, c_\updownarrow|V_{*0}|)$ and 
    $|v_{01}| < \min(c_\updownarrow|V_{*0}|/w_\uparrow, c_\leftrightarrow|V_{0*}|)$
    then $|c_1|(\frac{1}{2} - 6 \epsilon) \le |c_0| = 0$,
    for $c_1 \in \Ker \partial_1, c_1\text{ is locally minimal}$.
  Therefore, if $c_1 \ne 0$,
    at least one of 
      $|v_{10}| < \min(c_\leftrightarrow|V_{0*}|/w_\leftarrow, c_\updownarrow|V_{*0}|), 
      |v_{01}| < \min(c_\updownarrow|V_{*0}|/w_\uparrow, c_\leftrightarrow|V_{0*}|)$
    is violated.
  So $|c_1| \ge \min (c_\leftrightarrow|V_{0*}|/w_\leftarrow, c_\updownarrow|V_{*0}|, 
    c_\updownarrow|V_{*0}|/w_\uparrow, c_\leftrightarrow|V_{0*}|)$.
\end{proof}

\section{Proof of small set LTC lemma~\ref{lem:small-set-LTC} and found if short lemma~\ref{lem:found-if-short}} \label{sec:proof-key-lemmas}
Here, we prove the two key lemmas 
  for the linear distance and the linear time decoder.
We first prove the harder lemma for the linear time decoder,
  and obtain the small set LTC lemma as a corollary.
Alternatively, one can also prove the small set LTC lemma directly as in \cite{lin2022c}.

\subsection{Prove the found if short lemma~\ref{lem:found-if-short}}

Here is an overview of the proof.
For each $x_{00} \in V_{00}$,
  we set $n_{10}(x_{00}), n_{01}(x_{00})$
  according to the structure of lossless expanders.
We show that there exists $x_{00}$ such that the corresponding $n_{10}(x_{00}), n_{01}(x_{00})$ is flippable.
This is done through an averaging argument over $x_{00}$.
The averaged argument is shown by utilizing the inequalities from lossless expanders.

\begin{proof}[Proof of Lemma~\ref{lem:found-if-short}]
  
  By Corollary~\ref{cor:freely-implies-lossless},
    $(V_{00}, V_{10}, E_{*0})$ is a 1-sided $(c_\updownarrow |V_{0*}|/|V_{00}|, \epsilon_\updownarrow)$-lossless expander, 
    $(V_{01}, V_{11}, E_{*1})$ is a 1-sided $(c_\updownarrow |V_{0*}|/|V_{01}|, \epsilon_\updownarrow)$-lossless expander,
    $(V_{00}, V_{01}, E_{0*})$ is a 1-sided $(c_\leftrightarrow |V_{*0}|/|V_{00}|, \epsilon_\leftrightarrow)$-lossless expander,
    $(V_{10}, V_{11}, E_{1*})$ is a 1-sided $(c_\leftrightarrow |V_{*0}|/|V_{10}|, \epsilon_\leftrightarrow)$-lossless expander.

  Recall that $(v_{10}, v_{01})$ represents the error 
    and are assumed to be short enough
    so that the lossless lemmas apply.

  We first assign $n_{10}(x_{00}), n_{01}(x_{00})$
    for each $x_{00} \in V_{00}$.
  By Corollary~\ref{cor:lossless-split},
    there exist $n_{v_{10}}(x_{00}) \subseteq N_{v_{10}}(x_{00}), n_{v_{01}}(x_{00}) \subseteq N_{v_{01}}(x_{00})$,
    such that
    \begin{itemize}
      \item $\{n_{v_{10}}(x_{00}): x_{00} \in V_{00}\}$ forms a partition of $v_{10}$,
      \item $|N_{v_{10}}(x_{00}) - n_{v_{10}}(x_{00})| \le \epsilon_\leftrightarrow w_\rightarrow$,
      \item $\{n_{v_{01}}(x_{00}): x_{00} \in V_{00}\}$ forms a partition of $v_{01}$,
      \item $|N_{v_{01}}(x_{00}) - n_{v_{01}}(x_{00})| \le \epsilon_\updownarrow w_\downarrow$.
    \end{itemize}

  \begin{figure}
    \centering
    \begin{tikzpicture}
      \begin{scope}[shift={(-5,5)}]
        \def\eps{2pt}

        \fill (0,0) circle[radius=1pt] node [label=above left:$x_{00}$] {};
    
        \filldraw[fill=black, draw=black] (-\eps,-2) rectangle (\eps,-4);
        \filldraw[fill=gray, draw=black] (-\eps,-4) rectangle (\eps,-5);
        \filldraw[fill=white, draw=black] (-\eps,-5) rectangle (\eps,-8);
    
        \filldraw[fill=black, draw=black] (2,-\eps) rectangle (4,\eps);
        \filldraw[fill=gray, draw=black] (4,-\eps) rectangle (5,\eps);
        \filldraw[fill=white, draw=black] (5,-\eps) rectangle (8,\eps);
    
        \path (-\eps,-2) -- node[anchor=south, rotate = 90]{$n_{v_{10}}(x_{00})$} (-\eps,-4);
        \path (-\eps,-4) -- node[anchor=north, rotate = 90]{$N_{v_{10}}(x_{00}) - n_{v_{10}}(x_{00})$} (-\eps,-5);
        \path (-\eps,-5) -- node[anchor=south, rotate = 90]{$N_{10}(x_{00}) - N_{v_{10}}(x_{00})$} (-\eps,-8);
    
        \path (2,\eps) -- node[auto]{$n_{v_{01}}(x_{00})$} (4,\eps);
        \path (4,\eps) -- node[auto,swap]{$N_{v_{01}}(x_{00}) - n_{v_{01}}(x_{00})$} (5,\eps);
        \path (5,\eps) -- node[auto]{$N_{01}(x_{00}) - N_{v_{01}}(x_{00})$} (8,\eps);
      
      \end{scope}
      
      \draw (-3, -3) rectangle (3, 3);
      \begin{scope}[shift={(1pt,1pt)}, blue, rounded corners]
        \draw [save path = \AOI] (-1, 1) rectangle (3, 3)
              node [above] {$A_{01}(x_{00})$};
        \clip [use path = \AOI];
        \foreach\x in {-3, -2.5, ..., 3}{
          \draw [.!50!white] (\x, 1) -- +(2, 2);
        }
      \end{scope}
      
      \begin{scope}[shift={(1pt,1pt)}, blue, rounded corners]
        \draw [save path = \AIO] (-1, 1) rectangle (-3, -3)
              node [left] {$A_{10}(x_{00})$};
        \clip [use path = \AIO];
        \foreach\y in {3, 2.5, ..., -3}{
          \draw [.!50!white] (-1, \y) -- +(-2, -2);
        }
      \end{scope}
      
      \begin{scope}[orange, rounded corners]
        \draw [save path = \pathBOI] (-1, 1) rectangle (0, 3);
        \node [above] at (-1,3) {$B_{01}(x_{00})$};
        \clip [use path = \pathBOI];
        \foreach\x in {-3, -2.5, ..., 3}{
          \draw [.!50!white] (\x, 1) -- +(-2, 2);
        }
      \end{scope}
      
      \begin{scope}[orange, rounded corners]
        \draw [save path = \pathBIO] (-1, 1) rectangle (-3, 0);
        \node [left] at (-3,1) {$B_{10}(x_{00})$};
        \clip [use path = \pathBIO];
        \foreach\y in {3, 2.5, ..., -3}{
          \draw [.!50!white] (-1, \y) -- +(-2, 2);
        }
      \end{scope}
      
      \begin{scope}[teal, rounded corners]
        \draw [save path = \COI] (-1, 1) rectangle (3, 3);
        \path (-1,3) -- node [above] {$C_{01}(x_{00})$} (3,3);
        \clip [use path = \COI];
        \foreach\x in {-3.25, -2.75, ..., 3.25}{
          \foreach\y in {-3, -2.5, ..., 3}{
            \fill (\x,\y) circle[radius=1pt, .!50!white];
          }
        }
      \end{scope}
      
      \begin{scope}[teal, rounded corners]
        \draw [save path = \CIO] (-1, 1) rectangle (-3, -3);
        \path (-3,1) -- node [left] {$C_{10}(x_{00})$} (-3,-3);
        \clip [use path = \CIO];
        \foreach\x in {-3.25, -2.75, ..., 3.25}{
          \foreach\y in {-3, -2.5, ..., 3}{
            \fill (\x,\y) circle[radius=1pt, .!50!white];
          }
        }
      \end{scope}
      
      \begin{scope} [shift={(-1pt,-1pt)}, purple, rounded corners]
        \draw (0, 0) -- (0, 3) -- (-1, 3) -- (-1, 1)
              -- (-3, 1) -- (-3, 0) -- cycle
              node [below right] {$D(x_{00})$};
        \clip (0, 0) -- (0, 3) -- (-1, 3) -- (-1, 1)
              -- (-3, 1) -- (-3, 0) -- cycle;
        \foreach\x in{-3, -2.7, ..., 0}{
          \draw [.!50!white] (\x, 0) -- +(1, 3);
        }
      \end{scope}
      
    \end{tikzpicture}
    \caption{Sketch of the regions $A, B, C, D$ for the proof of Lemma~\ref{lem:found-if-short}.}
    \label{fig:local-view-small-set-LTC}
  \end{figure}
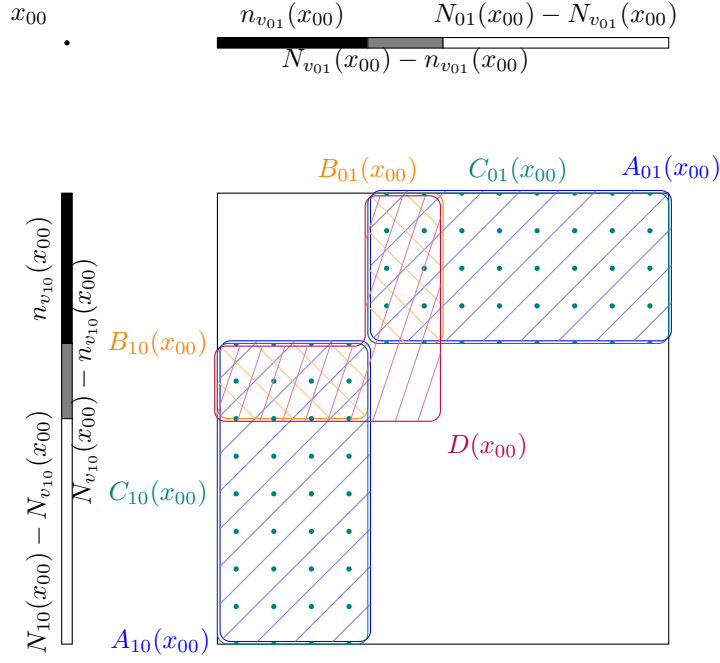


  Let $F(x_{00}) = \partial n_{v_{10}}(x_{00}) + \partial n_{v_{01}}(x_{00})$
    be the vertices (i.e. syndrome) that is flipped
    after flipping $n_{v_{10}}(x_{00})$ and $n_{v_{01}}(x_{00})$.
  Let $O(x_{00}) \subseteq F(x_{00})$ be the vertices in $F(x_{00})$
    that is has non trivial syndrome before the flip.
  Let $U(x_{00}) \subseteq F(x_{00})$ be the vertices in $F(x_{00})$
    that is unique neighbor to $v_{10} \cup v_{01}$.
  It is easy to see $U(x_{00}) \subseteq O(x_{00})$.
  
  To remind ourself, our goal is to show there exists $x_{00}$
    such that $|O(x_{00})| \ge \beta |F(x_{00})|$.
  Because $|O(x_{00})| \ge |U(x_{00})|$
    we can reduce the question to show
    there exists $x_{00}$, such that $|U(x_{00})| \ge \beta |F(x_{00})|$,
  To show the existence, we will show the averaged statement $\sum_{x_{00} \in V_{00}} |U(x_{00})| \ge \beta \sum_{x_{00} \in V_{00}} |F(x_{00})|$.
  
  Now, we begin to show the averaged statement.
  The method is to show that most elements in $F$ are also in $U$
    by removing the unwanted regions.
  Consider the following subsets 
    $A=\cup_{x_{00} \in V_{00}} A(x_{00})$, 
    $B=\cup_{x_{00} \in V_{00}} B(x_{00})$, 
    $C=\cup_{x_{00} \in V_{00}} C(x_{00})$, 
    $D=\cup_{x_{00} \in V_{00}} D(x_{00})$, 
  where $A$ is the main contribution and $B, C, D$ are error terms.

  $A(x_{00}) = A_{01}(x_{00}) \cup A_{10}(x_{00})$ is the region
    where the syndromes are flipped
    when we flip $n_{v_{10}}(x_{00})$ and $n_{v_{01}}(x_{00})$.
    This includes the possibility of flipping multiple times.

  \begin{equation}
    A_{01}(x_{00}) = \{(x_{00}, x_{10}, x_{01}, x_{10} \times_{x_{00}} x_{01}) : 
    x_{10} \in n_{v_{10}}(x_{00}), x_{10} \in N_{01}(x_{00})-n_{v_{01}}(x_{00})\},
  \end{equation}

  \begin{equation}
    A_{10}(x_{00}) = \{(x_{00}, x_{10}, x_{01}, x_{10} \times_{x_{00}} x_{01}) : 
    x_{10} \in N_{10}(x_{00})-n_{v_{10}}(x_{00}), x_{10} \in n_{v_{01}}(x_{00})\}.
  \end{equation}

  We mainly care about the last coordinate, i.e. the vertex in $V_{11}$.
    But we include the full information of the square
    in case $x_{10} \times_{x_{00}} x_{01}$ are not all distinct.
  Sometime we make this projection to the last coordinate implicit.

  From the definition, tt is easy to see $F(x_{00}) \subseteq A(x_{00})$.

  $B(x_{00}) = B_{01}(x_{00}) \cup B_{10}(x_{00})$ is a subregion of $A(x_{00})$.

  \begin{equation}
    B_{01}(x_{00}) = \{(x_{00}, x_{10}, x_{01}, x_{10} \times_{x_{00}} x_{01}) : 
    x_{10} \in n_{v_{10}}(x_{00}), x_{10} \in N_{v_{01}}(x_{00})-n_{v_{01}}(x_{00})\},
  \end{equation}

  \begin{equation}
    B_{10}(x_{00}) = \{(x_{00}, x_{10}, x_{01}, x_{10} \times_{x_{00}} x_{01}) : 
    x_{10} \in N_{v_{10}}(x_{00})-n_{v_{10}}(x_{00}), x_{10} \in n_{v_{01}}(x_{00})\}.
  \end{equation}

  $C(x_{00}) = C_{01}(x_{00}) \cup C_{10}(x_{00})$ is a subregion of $A(x_{00})$
    which has more than one neighbor in
      $n_{v_{10}}(x_{00})$ for $C_{01}(x_{00})$
      or $n_{v_{01}}(x_{00})$ for $C_{10}(x_{00})$.

  \begin{equation}
    C_{01}(x_{00}) = \{(x_{00}, x_{10}, x_{01}, x_{11}) \in A_{01}(x_{00}): 
      |N_{v_{10}}(x_{11})| > 1\},
  \end{equation}

  \begin{equation}
    C_{10}(x_{00}) = \{(x_{00}, x_{10}, x_{01}, x_{11}) \in A_{10}(x_{00}) : 
      |N_{v_{01}}(x_{11})| > 1\}.
  \end{equation}

  $D(x_{00}) = B(x_{00}) \cup D_{11}(x_{00})$ is a region slightly larger than $B(x_{00})$.

  \begin{equation}
    D_{11}(x_{00}) = \{(x_{00}, x_{10}, x_{01}, x_{10} \times_{x_{00}} x_{01}) : 
    x_{10} \in N_{v_{10}}(x_{00})-n_{v_{10}}(x_{00}), x_{10} \in N_{v_{01}}(x_{00})-n_{v_{01}}(x_{00})\}.
  \end{equation}

  Now, we claim the relation between the number of unique neighbor vertices
    and $A, B, C, D$.
  \begin{claim} \label{claim:unique-neighbor-and-ABCD}
    $|U| \ge |A| - |B| - |C| - 2 |D|$.
  \end{claim}

  We first assume the claim and prove the averaged statement.

  We here bound each of $|A|, |B|, |C|, |D|$.

  First, we study $|A|$.
  \begin{align*}
    |A_{01}(x_{00})| + |A_{10}(x_{00})| 
    &= |n_{v_{10}}(x_{00})| (w_\rightarrow - |n_{v_{01}}(x_{00})|)
    + (w_\downarrow - |n_{v_{10}}(x_{00})|) |n_{v_{01}}(x_{00})| \\
    &\ge (|n_{v_{10}}(x_{00})| w_\rightarrow + w_\downarrow |n_{v_{01}}(x_{00})|)/2,
  \end{align*}
  where the last inequality follows from local minimality
    $|n_{v_{10}}(x_{00})|/w_\downarrow + |n_{v_{01}}(x_{00})|/w_\rightarrow \le |N_{v_{10}}(x_{00})|/w_\downarrow + |N_{v_{01}}(x_{00})|/w_\rightarrow \le 1$.
  So
  \begin{equation} \label{eq:bound-A}
    |A| \ge (w_\rightarrow |v_{10}| + w_\downarrow |v_{01}|)/2.
  \end{equation}
  
  Next, we study $|B|$.
  \begin{align*}
    |B_{01}(x_{00})|
    &= |n_{v_{10}}(x_{00})| |N_{v_{01}}(x_{00})-n_{v_{01}}(x_{00})| \\
    &\le |n_{v_{10}}(x_{00})| \epsilon_\leftrightarrow w_\rightarrow.
  \end{align*}
  So
  \begin{equation*}
    |B| \le \epsilon_\leftrightarrow w_\rightarrow |v_{10}|
    + \epsilon_\updownarrow w_\downarrow |v_{10}|.
  \end{equation*}

  Then, we study $|C|$.
  $\sum_{x_{00} \in V_{00}} |C_{01}(x_{00})|$ are the number of edges that are connected to the non unique neighbors in $N_{11}(v_{10})$.
  From Lemma~\ref{lem:unique-expander}, we have $|N^{\textnormal{unique}}_{11}(v_{10})| \ge (1-2 \epsilon_\leftrightarrow) w_\rightarrow |v_{10}|$.
  This implies
    $\sum_{x_{00} \in V_{00}} |C_{01}(x_{00})| \le 2 \epsilon_\leftrightarrow w_\rightarrow |v_{10}|$,
  so
  \begin{equation*}
    |C| \le 2 \epsilon_\leftrightarrow w_\rightarrow |v_{10}|
    + 2 \epsilon_\updownarrow w_\downarrow |v_{10}|.
  \end{equation*}

  Finally, we study $|D|$.
  \begin{align*}
    |D_{11}|
    &= \sum_{x_{00} \in V_{00}} |N_{v_{10}}(x_{00})-n_{v_{10}}(x_{00})| |N_{v_{01}}(x_{00})-n_{v_{01}}(x_{00})| \\
    &\le \epsilon_\updownarrow w_\downarrow \epsilon_\leftrightarrow w_\rightarrow 
    \min(\lceil \frac{w_\uparrow}{\epsilon_\updownarrow w_\downarrow} |v_{10}| \rceil,
    \lceil \frac{w_\leftarrow}{\epsilon_\leftrightarrow w_\rightarrow} |v_{01}| \rceil) \\
    &\le 2 w_\uparrow \epsilon_\leftrightarrow (w_\rightarrow |v_{10}|),
  \end{align*}
  where the last inequality follows from $\lceil \frac{w_\uparrow}{\epsilon_\updownarrow w_\downarrow} |v_{10}| \rceil \le 2\frac{w_\uparrow}{\epsilon_\updownarrow w_\downarrow} |v_{10}|$.
  If $|v_{10}| = 0$, the inequality holds trivially.
  Otherwise, because $\epsilon_\updownarrow w_\downarrow \le w_\uparrow, |v_{10}| \ge 1$,
  we have $\frac{w_\uparrow}{\epsilon_\updownarrow w_\downarrow} |v_{10}| \ge 1$.
  Therefore,
  $\lceil \frac{w_\uparrow}{\epsilon_\updownarrow w_\downarrow} |v_{10}| \rceil \le 2 \frac{w_\uparrow}{\epsilon_\updownarrow w_\downarrow} |v_{10}|$.
  
  Combine with the result of $|B|$, we have
  \begin{equation*}
    |D| \le \epsilon w_\rightarrow |v_{10}|
    + \epsilon w_\downarrow |v_{10}|
    + 2 (w_\uparrow \epsilon_\leftrightarrow) (w_\rightarrow |v_{10}|).
  \end{equation*}

  Now, we combine the results and use
    $w_\uparrow \epsilon_\leftrightarrow \le \epsilon,
    \epsilon_\leftrightarrow \le \epsilon,
    \epsilon_\updownarrow  \le \epsilon$
  to obtain the desired result

  \begin{equation} \label{eq:bound-U-A}
    |U| \ge |A| - |B| - |C| - 2 |D| \ge (1-12\epsilon) |A| \ge (1-12\epsilon) |F|,
  \end{equation}
  where for the second inequality we use 
    $|A| \ge (w_\rightarrow |v_{10}| + w_\downarrow |v_{01}|)/2$,
    $|B| \le \epsilon (w_\rightarrow |v_{10}| + w_\downarrow |v_{01}|)$,
    $|C| \le 2 \epsilon (w_\rightarrow |v_{10}| + w_\downarrow |v_{01}|)$,
    $|D| \le 3 \epsilon (w_\rightarrow |v_{10}| + w_\downarrow |v_{01}|)$
    and for the last inequality we use
    $F(x_{00}) \subseteq A(x_{00})$.
  
  We suffice to prove the claim.
  The idea is to consider the elements $(x_{00}, x_{10}, x_{01}, x_{11})$ in $A-B-C$, 
    and show that if the $x_{11} \not \in U$,
    then such $x_{11}$ appears at most twice in $A-B-C$,
    and $x_{11}$ appears at least once in $D$.

  \begin{proof} [Proof of Claim~\ref{claim:unique-neighbor-and-ABCD}]
    Let $(x_{00}, x_{10}, x_{01}, x_{11}) \in A - B - C$.

    We first show $x_{11}$ appears at most twice in $A-B-C$.
    Suppose $(x_{00}, x_{10}, x_{01}, x_{11}) \in A_{01}(x_{00})$.
    Because $x_{10} \in n_{v_{10}}(x_{00})$, 
      we have $|N_{v_{10}}(x_{11})| \ge 1$.
    Because $(x_{00}, x_{10}, x_{01}, x_{11}) \not \in C$,
      we have $|N_{v_{10}}(x_{11})| \le 1$,
      so $|N_{v_{10}}(x_{11})| = 1$.

    This implies for $x_{11}$ to be in $A_{01}(x_{00}) - B - C$,
      there is a unique $x_{10} \in N_{v_{10}}(x_{11})$.
    Because $\{n_{v_{10}}(x_{00})\}_{x_{00} \in V_{00}}$ forms a partition of $v_{10}$,
      there is a unique choice of $x_{00}$.
    Therefore, there is a unique choice of tuple $(x_{00}, x_{10}, x_{01}, x_{11})$
      and $x_{11}$ appears at most once in $A_{01}(x_{00}) - B - C$.
    Similarly, $x_{11}$ appears at most once in $A_{10}(x_{00}) - B - C$,
      so $x_{11}$ appears at most twice in $A - B - C$.
    
    Now we show if $x_{11} \not \in U$,
      then $x_{11}$ appears at least once in $D$.
    Suppose $(x_{00}, x_{10}, x_{01}, x_{11}) \in A_{01}(x_{00})$.
    From the previous discussion, we have $|N_{v_{10}}(x_{11})| = 1$,
      and let $x_{10}$ be the unique element.
    We also know $x_{10} \in n_{v_{10}}(x_{00})$
      because $(x_{00}, x_{10}, x_{01}, x_{11}) \in A_{01}(x_{00})$
      and $x_{01} \in N_{01}(x_{00}) - N_{v_{01}}(x_{00})$
      because $(x_{00}, x_{10}, x_{01}, x_{11}) \not \in B_{01}(x_{00})$.

    When $x_{11} \not \in U$, $|N_{v_{10}}(x_{11})| + |N_{v_{01}}(x_{11})| \ge 2$, so we have $|N_{v_{01}}(x_{11})| \ge 1$.
    Pick any $x'_{01} \in N_{v_{01}}(x_{11})$.
    Then we have $(x'_{00} = x_{10} \times_{x_{11}} x'_{01}, x_{10}, x'_{01}, x_{11}) \in D$.
    First of all, $x'_{01} \in N_{v_{01}}(x_{11})$ but $x_{01} \not \in N_{v_{01}}(x_{11})$
      so $x'_{01} \ne x_{01}$.
    This implies $x'_{00} \ne x_{00}$,
      so $n_{v_{10}}(x_{00})$ is disjoint from $n_{v_{10}}(x'_{00})$.
    Because $x_{10} \in n_{v_{10}}(x_{00})$, we have $x_{10} \in N_{v_{10}}(x'_{00}) - n_{v_{10}}(x'_{00})$.
    Now, together with $x'_{01} \in N_{v_{01}}(x'_{00})$,
      we obtain $(x'_{00}, x_{10}, x'_{01}, x_{11}) \in D$.
  \end{proof}
\end{proof}

\subsection{Prove the small set LTC lemma~\ref{lem:small-set-LTC}}

Now, we can prove the small set LTC lemma as a simple corollary.
\begin{proof} [Proof of Lemma~\ref{lem:small-set-LTC}]
  Using the inequalities \ref{eq:bound-A} and \ref{eq:bound-U-A},
    we obtain $|c_0| \ge |U| \ge (1-12 \epsilon) |A|
    \ge (\frac{1}{2} - 6\epsilon) (w_\rightarrow |v_{10}| + w_\downarrow |v_{01}|)$.
\end{proof}

\end{document}

%% file: preamble.tex
\usepackage{amsthm,amsmath,amssymb}

\usepackage{authblk}

\usepackage{bm}
\usepackage{mathrsfs}
\usepackage{xcolor}

\usepackage{tikz,xcolor,graphicx}
\usetikzlibrary{patterns}

\usepackage{mathtools}

\usepackage{setspace}
\usepackage{empheq}

\usepackage{xargs}

\usepackage[unicode]{hyperref}
\usepackage{url}
\usepackage{hyperref,verbatim}

\usepackage{makecell}

\usepackage[textsize=small,textwidth=.8in]{todonotes}
\reversemarginpar

\definecolor{JLHgreen}{RGB}{50,100,50}

\newcommandx{\now}[2][1=]{\todo[inline, linecolor=red,backgroundcolor=red!25,bordercolor=red,#1]{#2}}

\newcommandx{\add}[2][1=]{\todo[inline, linecolor=teal,backgroundcolor=teal!25,bordercolor=teal,#1]{#2}}
\newcommandx{\improve}[2][1=]{\todo[inline, linecolor=violet,backgroundcolor=violet!25,bordercolor=violet,#1]{#2}}
\newcommandx{\change}[2][1=]{\todo[inline, linecolor=blue,backgroundcolor=blue!25,bordercolor=blue,#1]{#2}}
\newcommandx{\delete}[2][1=]{\todo[inline, linecolor=orange,backgroundcolor=orange!25,bordercolor=orange,#1]{#2}}

\usepackage{float}

\usepackage{cancel}
\usepackage{ulem}

\renewcommand{\emph}[1]{#1}

\usepackage{environ}
\NewEnviron{eqs}{%
  \begin{equation}\begin{split}
    \BODY
  \end{split}\end{equation}
}

\usepackage{color,xcolor}

\newtheorem{theorem}{Theorem}

\newtheorem{lemma}[theorem]{Lemma}
\newtheorem{claim}[theorem]{Claim}
\newtheorem{construction}[theorem]{Construction}
\newtheorem{corollary}[theorem]{Corollary}
\newtheorem{conjecture}[theorem]{Conjecture}

\newtheorem{definition}[theorem]{Definition}



%% file: macros.tex
\newcommand{\CC}{\mathbb{C}}

\newcommand{\FF}{\mathbb{F}}
\newcommand{\NN}{\mathbb{N}}
\newcommand{\RR}{\mathbb{R}}

\newcommand{\cC}{\mathcal{C}}
\DeclareMathOperator{\Ker}{ker}
\DeclareMathOperator{\Ima}{im}

\renewcommand{\>}{\rangle}

\def \a{\alpha}
\def \b{\beta}
\def \v{\nu}